\documentclass[11pt]{article}

\usepackage{algorithm}
\usepackage[noend]{algpseudocode}
\usepackage{graphicx,subfigure,amsmath,amsthm,amsfonts,bm,epsfig,epsf,url,dsfont}
\usepackage{fullpage}
\usepackage[small,bf]{caption}
\usepackage[margin=1in]{geometry}
\usepackage{bbm}      
\usepackage{booktabs}
\usepackage{cases}
\usepackage{multirow}
\usepackage{graphicx}
\usepackage{xspace}
\usepackage{color}
\usepackage{tikz}
\usepackage[colorlinks]{hyperref}     
\usepackage{cleveref}
\usepackage{mathtools}
\usepackage{newtxmath}      
\usepackage{verbatim}
\usetikzlibrary{shapes.misc}

\newtheorem{theorem}{Theorem}
\newtheorem{definition}[theorem]{Definition}
\newtheorem{lemma}[theorem]{Lemma}
\newtheorem{proposition}[theorem]{Proposition}
\newtheorem{corollary}[theorem]{Corollary}
\newtheorem{claim}[theorem]{Claim}

\newcommand{\cavg}{c_\textup{avg}}
\newcommand{\cD}{\mathcal D}

\newcommand{\cM}{\mathcal M}

\newcommand{\N}{\mathbb{N}}
\newcommand{\R}{\mathbb{R}}

\newcommand{\pesc}{p_\textup{esc}}

\newcommand{\Ind}{\vvmathbb 1}
\let\1\Ind

\DeclareMathOperator{\E}{\mathbb{E}}

\DeclareMathOperator{\diam}{diam}

\DeclareMathOperator{\lleft}{left}
\DeclareMathOperator{\rright}{right}

\DeclareMathOperator{\alg}{alg}
\DeclareMathOperator{\ONL}{ONL}

\DeclareMathOperator{\kSRV}{kSRV}
\DeclareMathOperator{\MTS}{MTS}
\DeclareMathOperator{\MSS}{MSS}
\DeclareMathOperator{\LGT}{LGT}
\DeclareMathOperator{\ETG}{ETG}
\DeclareMathOperator{\MA}{MA}
\DeclareMathOperator{\disPG}{DPG}
\DeclareMathOperator{\ekTX}{ekTX}
\DeclareMathOperator{\hkTX}{hkTX}

\DeclareMathOperator{\rand}{rand}
\DeclareMathOperator{\detr}{detr}
\DeclareMathOperator{\distr}{distr}
\DeclareMathOperator{\opt}{opt}
\DeclareMathOperator{\Lip}{Lip}

\usepackage{bm}

\begin{document}
	\setlength{\parskip}{1ex} 
\setlength{\parindent}{5ex}

\title{The Randomized $k$-Server Conjecture Is False!}

\author{S\'ebastien Bubeck\thanks{Microsoft Research, Redmond, WA, United States, 
                 {\tt sebubeck@microsoft.com}.} 
	\and
	Christian Coester\thanks{Department of Computer Science, University of Oxford, Oxford, 
	         United Kingdom, {\tt christian.coester@cs.ox.ac.uk}.}
	\and
	Yuval Rabani\thanks{The Rachel and Selim Benin School of Computer Science and Engineering, 
		The Hebrew University of Jerusalem, Jerusalem, Israel, {\tt yrabani@cs.huji.ac.il}.
		Research supported in part by ISF grants 3565-21 and 389-22, and by BSF grant 2018687.}
}

\date{\today}

\maketitle

\begin{abstract}
We prove a few new lower bounds on the randomized competitive ratio for the $k$-server problem and other related problems, resolving some long-standing conjectures. In particular, for metrical task systems (MTS) we asympotically settle the competitive ratio and obtain the first improvement to an existential lower bound since the introduction of the model 35 years ago (in 1987).

More concretely, we show:
\begin{enumerate}
\item There exist $(k+1)$-point metric spaces in which the randomized competitive ratio for the $k$-server problem is $\Omega(\log^2 k)$. This refutes the folklore conjecture (which is known to hold in some families of metrics) that in all metric spaces with at least $k+1$ points, the competitive ratio is $\Theta(\log k)$.
\item Consequently, there exist $n$-point metric spaces in which the randomized competitive ratio for MTS is $\Omega(\log^2 n)$. This matches the upper bound that holds for all metrics. The previously best existential lower bound was $\Omega(\log n)$ (which was known to be tight for some families of metrics).
\item For all $k<n\in\N$, for \emph{all} $n$-point metric spaces the randomized $k$-server competitive ratio is at least $\Omega(\log k)$, and consequently the randomized MTS competitive ratio is at least $\Omega(\log n)$. These universal lower bounds are asymptotically tight. The previous bounds were $\Omega(\log k/\log\log k)$ and $\Omega(\log n/\log \log n)$, respectively.
\item The randomized competitive ratio for the $w$-set metrical service systems problem, and its equivalent width-$w$ layered graph traversal problem, is $\Omega(w^2)$. This slightly improves the previous lower bound and matches the recently discovered upper bound.
\item Our results imply improved lower bounds for other problems like $k$-taxi, distributed paging, and metric allocation.
\end{enumerate}
These lower bounds share a common thread, and other than the third bound, also a common construction.
\end{abstract}

\thispagestyle{empty}
\newpage
\setcounter{page}{1}

\begin{flushright}
	\begin{tabular}{l}
		{\em Two roads diverged in a wood, and I---}\\
		{\em I took the one less traveled by,}\\
		{\em And that has made all the difference.}\\
		\ \ \ \ -- Robert Frost
	\end{tabular}
\end{flushright}

\section{Introduction}

Since its inception in~\cite{MMS88}, the $k$-server problem has been the driving 
challenge shaping research on competitive analysis of online algorithms, a computing
paradigm pioneered in~\cite{ST84}. The problem is simple to state. 
(See Section~\ref{sec: prelim} for the descriptions of all the problems discussed in the
introduction.) There are compellingly 
beautiful conjectures regarding the competitive ratio for deterministic and randomized 
algorithms. The attempts, partly successful, to prove these conjectures have led to the 
development of powerful tools whose impact reaches far beyond the $k$-server problem, 
or online computing for that matter. These include the application of and the contribution 
to concepts such as the work function, quasi-convexity, online choice of online algorithms, 
hierarchically separated trees (HSTs), non-contracting Lipschitz maps of metric spaces, metric 
Ramsey theory, the online primal-dual schema, entropic regularization, and the online 
mirror descent schema.

In this paper, we refute the randomized $k$-server conjecture, which states that the
randomized competitive ratio for the problem in {\em all} metric spaces (with at least
$k+1$ points) is $\Theta(\log k)$. We construct, for all $k\in\N$, metric spaces
on $k+1$ points in which the randomized competitive ratio for this problem
is $\Omega(\log^2 k)$. This also implies stronger lower bounds for a number
of related problems. In particular: ($i$) Metrical task systems were introduced 
in~\cite{BLS87}. The $k$-server problem in a $(k+1)$-point metric space is a special 
case of the metrical task systems problem in the same metric space. Therefore, we 
get that there are, for all $n\in\N$, $n$-point metric spaces in which the randomized 
competitive ratio for the metrical task systems problem is $\Omega(\log^2 n)$.
($ii$) The $k$-taxi problem was introduced in~\cite{FRR90}. It is trivially at least as
hard as the $k$-server problem in the same metric space. Hence, our $k$-server
lower bounds carry over to this problem.
($iii$) The distributed paging (a.k.a. constrained file allocation) problem was introduced 
in~\cite{BFR92}. In~\cite[Theorem 3.1]{ABF93} it was shown that the $k$-server problem 
in a $(k+1)$-point metric space is a special case of distributed paging in a network
inducing the same metric on $k+1$ processors, with total capacity of the caches at the
processors of $m$ pages, and $m-k+1$ distinct pages in the system. Thus, we get that
for all $n\in\N$, there exist $n$-processor networks in which the randomized competitive
ratio for the distributed paging problem is $\Omega(\log^2 n)$. ($iv$) The metric
allocation problem was introduced in~\cite{BC21b}. The same paper shows that the
randomized metrical task systems problem is a special case of this problem in the
same metric space. Hence, our results imply an $\Omega(\log^2 n)$ lower bound
for metric allocation, where $n$ is the cardinality of the underlying metric space.

Moreover, the ideas behind the construction of the lower bound examples are
inspired by a construction of lower bound examples for a different problem. Layered
graph traversal (a.k.a. metrical service systems) was introduced in~\cite{PY89}.
Here we show a lower bound of $\Omega(w^2)$ on the randomized competitive
ratio for traversing width-$w$ layered graphs. This improves upon the previous
lower bound of $\Omega(w^2/\log^{1+\epsilon} w)$ (for all $\epsilon > 0$) of~\cite{Ram93},
and matches asymptotically the recent upper bound of~\cite{BCR22}. It also implies
the same asymptotically tight lower bound on the randomized competitive ratio for
the depth-$w$ evolving tree game, on account of the reduction used in~\cite{BCR22}
and the matching upper bound in that paper.

The results for the $k$-server and related problems are quite surprising. 
The new $k$-server lower bounds are asymptotically 
tight for $k+1$ point metric spaces, and more generally this is true of our 
metrical task systems lower bounds. 
Matching upper bounds appear in~\cite{BCLL19,CL19}. Previously, it has been widely
conjectured 
(see~\cite{FKLMSY91,KMMO90,KRR91,BKRS92,Sei01,Kou09,BBN10a,BBMN11,Lee18,BBCJ20,HZ22})
that in all metric spaces on more than $k$ points, the randomized 
competitive ratio for the $k$-server problem is $\Theta(\log k)$. This was also 
established in some special cases~\cite{FKLMSY91,FM00,BBN07}. The best known upper bounds for general metrics are $O(\log^2 k\log n)$ and $O(\log^3 k\log\Delta)$~\cite{BCLLM18}, where $n$ is the number of points in the metric space and $\Delta$ is the ratio between the largest and smallest non-zero distance, and $O(k)$ by using a deterministic algorithm~\cite{KP94} when both $n$ and $\Delta$ are very large.
Moreover, the deterministic $k$-server conjecture, which states that in every metric 
space on at least $k+1$ points the deterministic competitive ratio is exactly $k$, is
nearly resolved. A lower bound of $k$ is known to hold in all metrics~\cite{MMS88}.
It is tight in many special cases~\cite{ST84,MMS88,CKPV91,CL91,KP96,BCL02,CK21,HZ22}, and it is within
a factor of $2-\frac 1 k$ of the truth in all cases~\cite{KP94}. Thus, it was unexpected 
that in the randomized case the competitive ratio would vary widely among different 
metric spaces. Similarly, for metrical task systems, in some metric spaces an
asymptotically tight $\Theta(\log n)$ bound is known. In fact, the upper bound holds 
in all HST metrics~\cite{BCLL19}; see also previous work 
in~\cite{BLS87,BBBT97,IS95,Sei99,FM00}. It has been conjectured that this is the truth 
in all metric spaces (similarly to the deterministic case, where all metric spaces are known 
to be equally hard with a competitive ratio of exactly $2n-1$~\cite{BLS87}). Obviously, our 
results refute this conjecture. The lower bounds for $k$-taxi, for distributed paging, and for 
metric allocation are not matched by any known universal upper bound. However, in uniform 
metric spaces tight bounds of $\Theta(\log n)$ for the latter two problems~\cite{ABF93,BC21b} and $\Theta(\log k)$ for $k$-taxi~\cite{C22} are known (for metric 
allocation and $k$-taxi this holds even in a somewhat more general setting of weighted star metrics).

\paragraph{Existential and universal lower bounds.} The previous best existential lower bound for metrical task systems, $\Omega(\log n)$, was already proved in 1987 when the model was introduced \cite{BLS87}. It uses a coupon collector argument, and the analogous proof had also remained the previously best existential lower bound for the $k$-server problem. In contrast, the best known \emph{universal} lower bounds (i.e., lower bounds that hold for \emph{all} metrics) had developed over time: Initially it seemed plausible that there might even exist metric spaces with constant competitive ratio, until a first super-constant universal lower bound was shown in~\cite{KRR91}. This was based on the idea of showing that every metric space contains a large subspace belonging to a family for which a super-constant lower bound is known. This idea was developed further 
in~\cite{BKRS92,BBM01,BLMN03}, eventually resulting in lower bounds of $\Omega(\log n / \log \log n)$ and $\Omega(\log k / \log \log k)$, respectively, for metrical task systems and the $k$-server problem for all metric spaces of $n>k$ points. These lower bounds are implied by an $\Omega(\log n)$ lower bound on metrical task systems in $\Omega(\log^2 n)$-HST metrics from~\cite{BBM01}, in combination with lower bounds 
on the size of a subspace close to such a structure in any metric space from~\cite{BLMN03}. Here, we close the remaining gap to the upper bounds of $O(\log n)$ and $O(\log k)$ that hold in some metrics by proving universal lower bounds of $\Omega(\log n)$ and $\Omega(\log k)$, respectively. We show that such lower bounds hold for all $1$-HST
metrics. The universal lower bounds are then implied from~\cite{BLMN03}.

\paragraph{Overview of the existential lower bound proof.} 
Our existential lower bounds use a construction of a metric space akin to the diamond graph 
used in~\cite{NR02,BC03,LN04,MN07,ACNN11} to prove lower bounds on bi-Lipschitz distortion of
metric embeddings and dimension reduction in $\ell_1$, and in~\cite{IW91} to prove
lower bounds on the online Steiner tree problem. The main idea is demonstrated
by the following basic construction that is used in Section~\ref{sec: basic} to illustrate a
somewhat weaker, yet simpler, lower bound than our main result. We use the shortest
path metric on the nodes of an inductively constructed graph. Let $w\in\N$. The $w$
graph is constructed as follows. Take a cycle whose length is a multiple of $6$. Let $s$
and $t$ be two antipodal points in this cycle. They partition the cycle into two $s$-$t$
paths whose length is a multiple of $3$. We'll consider the first, second, and last third
of both paths. But first, we replace every edge in this graph with a copy of the $w-1$
graph, identifying the endpoints of the edge with the two chosen terminals of the
recursive graph. Note that in an $n$-point metric space, if $k=n-1$ then the movement 
of the servers can be defined alternatively as the movement of the ``hole," a.k.a. 
anti-server---the unique position that no server occupies. The anti-server can be forced 
to choose a location from a subset of the points by requesting repeatedly the other points.

Now, we describe the bad sequence of requests, informally. Our argument
uses Yao's minimax principle---we design a random request sequence that drives the 
cost of every deterministic algorithm to reach or exceed the lower bound, which we'll
denote here by $C_w$. This is done in three stages. The first stage makes sure the 
anti-server traverses the first third of one of the two $s$-$t$ paths (which now consist 
of chains of $w-1$ graphs), while paying a factor of $C_{w-1}$ over the shortest path
to this goal. The main purpose of this stage is to create two targets to choose from
which are far apart. The last third is a mirror image of the first third and is intended
to ensure that the anti-server reaches $t$, while paying the same factor $C_{w-1}$
over the shortest path to that goal. The middle third is where the increase in the
competitive ratio happens. There we repeatedly, independently, and with equal
probability choose one of the two paths, and force the anti-server either to move
forward on the chosen path, or to stay put on the other path. This is repeated
for the length of the middle third, so in expectation, wherever the anti-server 
lies, it has to advance half the length of this interval. The initial distance
gained from $s$ guarantees that switching paths does not save any cost.
Moreover, in expectation one of the two options ``suffered" fewer ``hits,"
roughly square root of the number of repetitions fewer. We now take advantage
of this expected gap, and generate more requests that force the anti-server
to use ``the road less traveled by" and move there all the way to $t$. The
excess steps force $C_w$ to be at least $C_{w-1} + \Theta(\sqrt{C_{w-1}})$,
and this shows an $\Omega(w^2)$ lower bound. Finally, the size of this graph
is $\exp(w\log w)$, so this argument gives $\Omega((\log k/\log\log k)^2)$.
To improve the latter bound, we need to ``compress" the graph to use cycles
of length $6$, recursively, so its size is $\exp(w)$. This complicates the
argument considerably. As the graph is symmetric with respect to the roles of 
$s$ and $t$, we can then reverse the process (now going from $t$ to $s$) and 
thus repeat it as many times as desired.

We note that our constructed metric spaces require distortion $\Omega(\log n)$
to represent as a convex combination of HSTs (where $n$ is the number of
points in the metric space). This is obvious from the fact that these metrics
contain a path of $n^{\epsilon}$ equally spaced points, for some constant
$\epsilon > 0$. In fact, our lower bounds illustrate the somewhat surprising 
conclusion that the approximation of metric spaces as convex combinations
of HSTs (from~\cite{FRT04}), which underlies the tight universal upper bounds 
in~\cite{BCLL19}, cannot be circumvented.

\paragraph{Organization.} The rest of the paper is organized as follows. In Section~\ref{sec: prelim} we
define the problems we show lower bounds for, and discuss the reductions among
them. In Section~\ref{sec: basic} we prove the basic construction outlined above, and also prove the lower bound of $\Omega(w^2)$ for layered graph traversal.
In Section~\ref{sec: main} we prove the stronger bound for $\MTS$ and $k$-server, our main result. Finally,
in Section~\ref{sec: universal} we prove our improved universal lower bound.

\section{Preliminaries}\label{sec: prelim}

An online (minimization) problem is a two-player game between an adversary and an online 
algorithm. Keeping the discussion somewhat informal, the adversary chooses a finite sequence 
of requests $\rho = \rho[1]\rho[2]\rho[3]\dots$. The algorithm chooses a response
function $\alg$ that maps every request sequence to an action. When the game is played,
the algorithm responds to each request $\rho[i]$ with 
$\alg[i] = \alg(\rho[1]\rho[2]\dots\rho[i])$, thus generating a response sequence
$\alg = \alg[1]\alg[2]\alg[3]\dots$ (abusing notation slightly). A randomized algorithm $\widetilde{\alg}$
chooses $\alg$ from a probability distribution on such functions. Thus, a play of the game is 
marked by a pair of equal-length sequences $(\rho,\alg)$. For each pair of an algorithm and request sequence there is an associated cost. We denote by $c_{\alg}(\rho)$ the cost of algorithm $\alg$ on request sequence $\rho$, and by $c_{\opt}(\rho)=\min_{\alg} \{c_{\alg}(\rho)\}$ the cost of the optimal (offline) algorithm.
\begin{definition}\label{def:CR}
	The randomized competitive ratio of an online problem $\ONL$, denoted $C_{\rand}^{\ONL}$, is 
	the infimum over all $C$ that satisfy the following condition. There exists a randomized algorithm
	$\widetilde{\alg}$ and some constant $\kappa$ such that for every adversary strategy $\rho$,
	$$
	\E[c_{\alg}(\rho):\ \alg\sim\widetilde{\alg}] \le C\cdot c_{\opt}(\rho)+\kappa.
	$$
\end{definition}
Unless stated otherwise, the constant $\kappa$ can be arbitrary. However, for some problems it is more typical to require $\kappa=0$.

\begin{definition}
	The distributional lower bound on the randomized competitive ratio of an online problem $\ONL$,
	denoted $C_{\distr}^{\ONL}$, is the supremum over all $C$ that satisfy the following condition. For all $\kappa$ (which are allowed in Definition~\ref{def:CR}), there
	exists a probability distribution $\tilde{\rho}$ over adversary strategies such that for every (deterministic) algorithm
	$\alg$,
	$$
	\E[c_{\alg}(\rho):\ \rho\sim\tilde{\rho}]\ge C\cdot \E[c_{\opt}(\rho):\ \rho\sim\tilde{\rho}]+\kappa.
	$$
\end{definition}

The following theorem is known as Yao's minimax principle. Notice that it cannot be deduced trivially
from von Neumann's minimax principle by setting $\frac{c_{\alg}(\rho)-\kappa}{c_{\opt}(\rho)}$ to
be the zero-sum value of a play $(\rho,\alg)$, because the distributional lower bound is about
the ratio of expectations, not the expected ratio. However, a similar LP duality argument yields it.\footnote{One may prove it for a fixed $\kappa$ first, and then take the infimum over all $\kappa$.}
\begin{theorem}[Yao's minimax]\label{thm: Yao minimax}
	$C_{\rand}^{\ONL} = C_{\distr}^{\ONL}$.
\end{theorem}

\paragraph{Some notation.}
Let $\rho$ be a request sequence. We will use subscripts $\rho_i$ to indicate subsequences in
an underlying partition $\rho = \rho_1\rho_2\dots\rho_m$. Note that each $\rho_i$ is a sequence,
not necessarily a single request. We denote by $\rho_{\le i}$ the prefix $\rho_1\rho_2\dots\rho_i$. 
Also, for notational convenience, $\rho_{\le 0}$ stands for the empty sequence. For two indices 
$i\le j$, we denote by $\rho_{i\le\cdot\le j}$ the subsequence $\rho_i\rho_{i+1}\dots\rho_j$.

\paragraph{Escape prices.} 
Given an online problem $\ONL$, we will sometimes consider an {\em escape price relaxation} of $\ONL$, 
defined as follows. There is an underlying escape price parameter $p\ge 0$. At certain points in the request sequence 
(potentially all of them), the online algorithm is allowed as an additional option to respond to a request by bailing out of the game. If the online algorithm responds to a request $\rho[t]$ by bailing out of the game, then the escape price $p$ is added to its cost, but the online algorithm pays no additional cost for any requests thereafter (including the request $\rho[t]$ itself). If we allow this additional option of bailing out only in response to requests $\rho[t]$ belonging to some specific subsequence $\rho_{i\le\cdot\le j}$, then we say that the escape price is available on $\rho_{i\le\cdot\le j}$.

In the escape cost relaxation of a problem, we allow invoking the escape price only to the online algorithm, but not to the (offline) algorithms defining the base cost. Thus, the escape price option only helps the online player, and the competitive ratio for the escape price relaxation of $\ONL$ is at most
$C_{\rand}^{\ONL}$, regardless of the value of $p$.\footnote{The escape price relaxation is reminiscent of a combination of 
the original problem and the ski rental problem, where the classical cost corresponds to rental price and $p$ is the 
buying price.}

\paragraph{More notation.}
In some online problems, there is an adjustable starting configuration $s$ which affects the cost. Given a play 
$(\rho,\alg)$ with starting configuration $s$, we often write $c_{\alg,s}(\rho)$ to denote the cost of the 
algorithm in this case, and $c_{\opt,s}(\rho)$ for the corresponding optimal (offline) cost. Moreover, let $\rho = \rho_1\rho_2$ be the concatenation 
of two request sequences $\rho_1$ and $\rho_2$. For a play $(\rho,\alg)$ with starting configuration $s$, we 
write
$$
c_{\alg,s}(\rho_2\mid \rho_1) := c_{\alg,s}(\rho) - c_{\alg,s}(\rho_1),
$$
denoting the partial cost of the online algorithm on the subsequence $\rho_2$ when serving the request 
sequence $\rho = \rho_1\rho_2$. When the initial location $s$ is clear from the context, we 
will often drop it from the notation.

Now let us discuss the online problems considered in this paper. All the problems considered
in this paper are defined in the context of a metric space. See Appendix~\ref{sec: metrics}
for some basic definitions and notation.

\paragraph{\boldmath $k$-server.}
There is an underlying metric space ${\cal M} = (M,d)$. The adversary's strategy $\rho$ is
a sequence of points $\rho[1],\rho[2],\dots\in M$. The algorithm controls $k$ identical 
servers, initially located at $k$ distinct points in $M$. We may assume w.l.o.g. that $|M| > k$,
otherwise the problem is trivial. The response $\alg[i]$ to a request $\rho[i]$ moves one
of the servers from its current location to $\rho[i]$. This adds to the cost of the algorithm the
distance travelled by the server. We denote this problem as $\kSRV$.

\paragraph{Metrical task systems.}
Here, too, there is an underlying metric space ${\cal M} = (M,d)$, which must be finite.
The elements of $M$ are called {\em states}.
An adversary's request is a vector in $(\R_+\cup\{\infty\})^M$, where $\R_+$ denotes
the set of non-negative real numbers. The algorithm begins at an arbitrary $\alg[0]\in M$.
In response to $\rho[i]\in (\R_+\cup\{\infty\})^M$, the algorithm must choose a state
$\alg[i]\in M$. This adds to its cost $d(\alg[i-1],\alg[i]) + \rho[i](\alg[i])$. We 
denote this problem as $\MTS$.

\paragraph{\boldmath $k$-taxi.}
The setting is identical to that of the $k$-server problem. The adversary's strategy
consists of a sequence of pairs of points in $M$. In response to a request, the
algorithm must move a server to the first point in the pair, and then move the
same server from the first point to the second point. There are two flavors to this
problem, differing in the definition of the cost. In {\em easy} $k$-taxi, the algorithm
pays for the entire move. We denote this flavor by $\ekTX$. In {\em hard} $k$-taxi, 
the algorithm pays only for the move to the first point in the request and not for the 
move to the second point (which of course changes the location of the server towards 
the following requests). We denote this flavor by $\hkTX$.
\begin{proposition}[folklore]
For all metric spaces ${\cal M} = (M,d)$ and for all finite $k < |M|$, 
$C_{\rand}^{\hkTX}({\cal M})\ge C_{\rand}^{\ekTX}({\cal M})\ge 
C_{\rand}^{\kSRV}({\cal M})$.\footnote{We will use the notation 
$C_{\rand}^{\ONL}(P)$ (also $C_{\detr}^{\ONL}(P)$ for the deterministic competitive ratio) to specify 
the setting $P$ in which the problem $\ONL$ is considered.} 
\end{proposition}

We note that there are metric spaces ${\cal M}$ for which
$C_{\detr}^{\hkTX}({\cal M})\gg C_{\detr}^{\kSRV}({\cal M})$, see~\cite{CK19}.
On the other hand, in all settings, $C_{\detr}^{\ekTX}({\cal M})\le C_{\detr}^{\kSRV}({\cal M})+2$,
see~\cite{Kos96}.

\paragraph{Distributed paging.}
There is a network of processors. The communication links and delays induce
a metric on the set of processors, so we regard the setting as a metric space
${\cal M} = (M,d)$, where $M$ is the finite set of processors. Each processor
$x\in M$ is endowed with a cache that can hold $m_x\in\N$ memory pages. Let
$m = \sum_{x\in M} m_x$. The network in its entirety holds $f\le m$ distinct
pages. They can be replicated in different processors, but for each distinct
page at least one copy must be held somewhere in the network at all times.
The adversary requests pairs $(p,x)$, where $p$ is one of the $f$ pages,
and $x\in M$ is a processor. In response, the algorithm must bring a copy
of $p$ to $x$. If the page is not already there, the algorithm can replicate 
a copy at another processor, or move that copy, to $x$. Either way, this costs
the distance to the other processor. If $x$'s cache is full, the algorithm must
evict a page to make room for $p$, and the evicted page can be discarded, or
(this is necessary if it is the last copy) moved to a vacant slot in another processor,
incurring a cost equal to the distance between the processors. We denote this 
problem as $\disPG$. The following proposition is a 
special case of~\cite[Theorem 3.1]{ABF93}.
\begin{proposition}[Awerbuch, Bartal and Fiat \cite{ABF93}]
	Consider any $n$-node network ${\cal M}$. Let $k = n-1$, and let
	$f = m-n+2$. Then, $C_{\rand}^{\disPG}({\cal M},m,f) = 
	\Omega\left(C_{\rand}^{\kSRV}({\cal M})\right)$.
\end{proposition}

\paragraph{Metric allocation.}
As usual, there is an underlying finite metric space ${\cal M} = (M,d)$. The algorithm
maintains a fractional allocation of a resource to the points on $M$, denoted by a vector 
$p\in\R_+^M$ with $\sum_{x\in M} p_x = 1$. Moving from one allocation
to another adds to the algorithm's cost the transportation (a.k.a. earthmover) distance 
between the two, under $d$. Each request of the adversary is defined by assigning to
every $x\in M$ a non-increasing convex function $\phi_x: [0,1]\rightarrow\R_+$. If
the algorithm serves the request using the allocation $p$, then this adds
to its cost $\sum_{x\in M} \phi_x(p_x)$. We denote this problem by $\MA$.
\begin{proposition}[Bansal and Coester~\cite{BC21b}]
	For all finite metric spaces ${\cal M}$, $C_{\rand}^{\MA}({\cal M})=C_{\detr}^{\MA}({\cal M})\ge C_{\rand}^{\MTS}({\cal M})$.
\end{proposition}

\paragraph{Layered graph traversal.}
Here the adversary selects a layered graph $G=(V,E)$, with edge lengths 
$L:E\rightarrow\R_+$, first layer $\{s\}$, last layer $\{t\}$, and edges only between vertices of adjacent layers. The graph
is presented to the algorithm layer by layer. When a layer is presented, all
the edges to the previous layers and their lengths are revealed. Starting at
$s$, in order to reveal a new layer, the algorithm must reach the previous 
layer (so at the start the second layer is revealed). The game ends when the
algorithm reaches $t$. The number of layers is not known until $t$ is revealed.
The cost of the algorithm is the total length of its traversed path in $G$. 
We denote this problem by $\LGT$. Let $w$ denote the maximum 
number of nodes within a layer of $G$.
\begin{proposition}[folklore]
	Let ${\cal M}$ be a finite metric space. Then, 
	$C_{\rand}^{\LGT}(w)\ge C_{\rand}^{\MTS}({\cal M})$,
	where $w$ is the number of points in ${\cal M}$.
\end{proposition}
\begin{proof}[Proof idea]
	When cost vector $\rho[i]$ is revealed in MTS, construct a new layer $L_i$ with a vertex for each point in $\cM$, and the edge between any $x\in L_{i-1}$ and $y\in L_i$ has length $d(x,y)+\rho[i](y)$.
\end{proof}
We note that the inequality here is known to be far from tight.

\paragraph{Small set chasing.}
There is an underlying metric space ${\cal M} = (M,d)$, not necessarily finite,
and a starting point $x_0\in M$. The adversary presents requests which are 
subsets of $M$, each of cardinality at most $w$. To serve a request, the 
algorithm must move to one of the points in the finite subset, incurring a
cost equal to the distance traversed. We denote this problem by $\MSS$.\footnote{The abbreviation stands for \emph{metrical service systems}, which is the historical name of the problem.}
The following propositions are known.
\begin{proposition}[folklore]\label{pr: MTS-kSRV-MSS}
	Let ${\cal M}$ be an $n$-point metric space and let $k=n-1$.
	Then, for every $w\in\{1,2,\dots,k\}$,
	$C_{\rand}^{\MTS}({\cal M})\ge C_{\rand}^{\kSRV}({\cal M})\ge
	C_{\rand}^{\MSS}({\cal M},w)$.
	The second inequality is an equality if $w=k$.
\end{proposition}
\begin{proof}[Proof idea]
	In $(n-1)$-server there is exactly one point not covered by a server, and this point corresponds to the server location in MSS and MTS. Then a request to a set $S\subset M$ in MSS corresponds to repeated $k$-server requests to all the points in $M\setminus S$. A $k$-server request at a point $p\in M$ corresponds to the MTS cost vector that assigns cost $\infty$ to $p$ and cost $0$ to other points.
\end{proof}

\begin{proposition}[Fiat, Foster, Karloff, Rabani, Ravid, Vishwanathan~\cite{FFKRRV91}]\label{prop:MSSLGT}
	Let ${\cal M}$ be any metric space, and let ${\cal U}$ be the Urysohn universal
	metric space. Then, for every $w\in\N$,
	$C_{\rand}^{\MSS}({\cal U},w)\ge C_{\rand}^{\LGT}(w)\ge C_{\rand}^{\MSS}({\cal M},w)$.
\end{proposition}

\paragraph{Evolving tree game.}
The adversary maintains a rooted tree $T$ of maximum depth $w$ and
non-negative edge lengths. The root $r$ is fixed throughout the game,
and always has one child. Initially, $T$ has a single edge of length $0$.
At each round of the game, the adversary can choose one of three types
of moves: (a) Pick a non-root leaf and increase the length of the edge
incident on it. (b) Pick a leaf other than the root and its child and delete it, 
and if the parent's degree drops to $2$ merge the two edges. (c) Create two
or more new nodes and attach them with edges of length $0$ to an existing
leaf. The algorithm must occupy a leaf at all times. Moving between nodes
adds to the algorithm's cost the length of the path connecting them. Staying
at a leaf while the adversary increases its incident edge length adds to the cost of the algorithm
the increase in edge length. We denote
this problem by $\ETG$.
\begin{proposition}[Bubeck, Coester, Rabani~\cite{BCR22}]
	$C_{\rand}^{\ETG}(w)\ge C_{\rand}^{\LGT}(w)$.
\end{proposition}
Our lower bound for LGT in this paper, along with the upper bounds in~\cite{BCR22},
establish that the competitive ratio of LGT is actually the same (up to constant factors) as that of the version of ETG where the tree is binary (or of constantly bounded degree) at all times.

\section{Basic Construction and Analysis}\label{sec: basic}

The main goal of this section is to describe, in a somewhat informal style, a simple construction 
and analysis which yields a lower bound of $\Omega\left(\left( \frac{\log n}{\log\log n } \right)^2\right)$ 
for $\MSS$ in $n$-point metric spaces with request sets of arbitrary size. Notice that this lower bound implies
immediately the same lower bound for $\MTS$ and for $\kSRV$, putting $k=n-1$,
on account of Proposition~\ref{pr: MTS-kSRV-MSS}. In Section~\ref{sec:LGT} we show that this lower bound can also be achieved with request sets of size at most $w=O\left(\frac{\log n}{\log \log n}\right)$, which implies the lower bound of $\Omega(w^2)$ for LGT on account of Proposition~\ref{prop:MSSLGT}.

\subsection{The family of metric spaces}

First we describe our family of metric spaces by induction. Let $m_0 \le m_1 \le m_2 \le \cdots$ 
be a sequence of natural numbers, to be determined later. We construct a sequence of finite
metric spaces ${\cal M}_0 = (M_0,d_0)$, ${\cal M}_1 = (M_1,d_1)$, ${\cal M}_2 = (M_2,d_2)$, 
$\dots$ of growing size as follows. Each metric in this sequence is the shortest paths metric of 
an underlying graph. The base case ${\cal M}_0$ is a single edge of weight $1$. Next, 
${\cal M}_1$ is a cycle with $6 m_0$ edges. In other words, we form a cycle of $6 m_0$ copies 
of ${\cal M}_0$. We also choose two special antipodal vertices/points $s$ and $t$, so 
$\diam({\cal M}_1) = d_1(s,t)$.

More generally, ${\cal M}_{w+1}$ is built from ${\cal M}_w$ in the same way as ${\cal M}_1$ was 
built from ${\cal M}_0$ (see Figure~\ref{fig:basic}). Namely, $M_{w+1}$ is a ``cycle" made of $6 m_w$ copies of ${\cal M}_w$. 
Slightly more precisely, consider a cycle with $6 m_w$ edges, where each edge $\{u,v\}$ is replaced 
by putting a copy of ${\cal M}_w$ with the two special vertices at $u$ and $v$. In particular, in 
${\cal M}_{w+1}$, each special vertex of one copy of ${\cal M}_w$ is identified with another special 
vertex in another copy of ${\cal M}_w$. In ${\cal M}_{w+1}$, the special vertices $s$ and $t$ are special 
vertices in two distinct copies of ${\cal M}_w$ such that $\diam({\cal M}_{w+1}) = d_{w+1}(s,t)$. Notice 
that ${\cal M}_{w+1}$ can be viewed as consisting of a {\em left path} of $3 m_w$ copies of ${\cal M}_w$ 
and a {\em right path} of another $3m_w$ copies of ${\cal M}_w$. For two subsets of points 
$S, S'\subset M_w$, we denote by $(S,S',i,j)$ the subset of $M_{w+1}$ that is made of the union of the 
set $S$ in the $i^{th}$ copy of ${\cal M}_w$ on the left path, and the set $S'$ in the $j^{th}$ copy of 
${\cal M}_w$ on the right path.

\begin{figure}
	\begin{center}
		\includegraphics[width=0.35\textwidth]{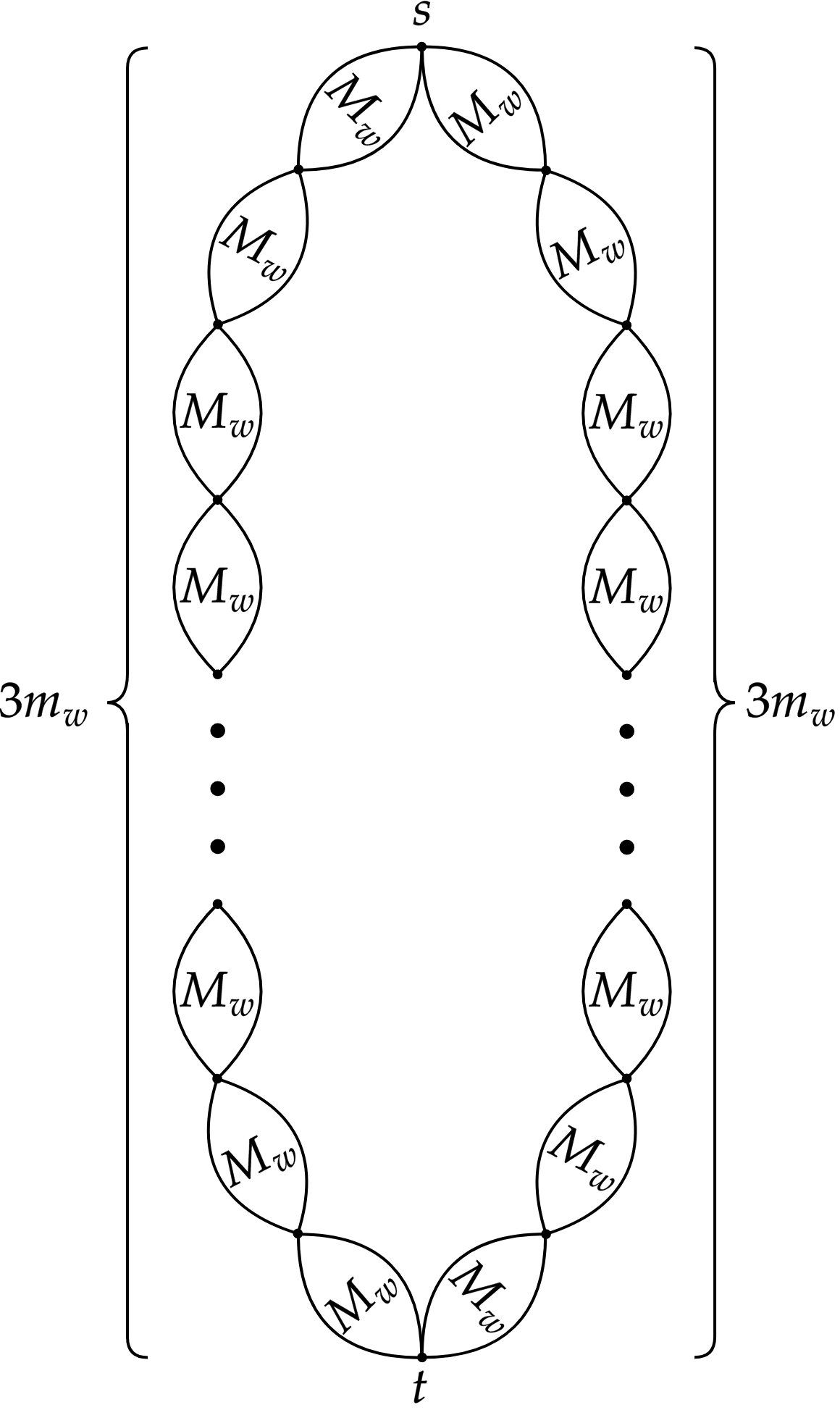}
	\end{center}
	\caption{Construction of $\cM_{w+1}$ in the $\Omega\left(\left(\frac{\log n}{\log \log n}\right)^2\right)$ lower bound.}\label{fig:basic}
\end{figure}

\subsection{The hard sequence}

Now let us describe by induction a hard random sequence $\rho^w$ of subsets of $M_w$. 
Recall that $\rho^w[j]$ denotes the $j^{th}$ request in the sequence $\rho^w$. This is a random 
subset of $M_w$. Also let $T_w$ denote the length of $\rho^w$. The construction gives $T_w$ 
as a random variable, but for simplicity we ignore this aspect as we can always pad a 
sequence with some dummy sets to attain a fixed length. We will always have $\rho^w[1]$ be 
the special vertex $s$ (the source) and $\rho^w[T_w]$ be the special vertex 
$t$ (the target). Moreover, our construction will guarantee that the request sequence can be 
satisfied by following a shortest path connecting $s$ and $t$ (so $c_{\opt}(\rho^w)$ is simply 
the length of this path). We denote by $\rho^{w,1}, \rho^{w,2}, \hdots$ a sequence of i.i.d. copies 
of $\rho^w$. The inductive construction of $\rho^{w+1}$ uses $\rho^w$ and proceeds in three stages.

\paragraph{Stage 1:} 
For every $i = 1, \hdots, m_{w}$ and $j \in \{1,2,\dots,T_w\}$, let 
\[
\rho^{w+1}[(i-1) T_w + j] = (\rho^{w,i}[j], \rho^{w,i}[j], i, i) \,. 
\]
In words, we traverse simultaneously the first third of the left and right path of ${\cal M}_{w+1}$, one copy of
${\cal M}_w$ at a time, using the hard sequence for ${\cal M}_w$.

\paragraph{Stage 2:} 
For every $i= 1, \hdots, m_{w}$, let $\epsilon_i \in \{0,1\}$ be a Bernoulli random variable independent of 
everything else. Let $\lleft(i) = m_w + \epsilon_1 + \hdots + \epsilon_i$, and let 
$\rright(i) = m_w + (1-\epsilon_1) + \hdots + (1-\epsilon_i)$. Now, for every $j \in \{1,2,\dots,T_w\}$, let
\[
\rho^{w+1}[(m_w + i-1) T_w + j] = 
(\rho^{w,m_w+i}[\epsilon_i j + (1-\epsilon_i) T_w], \rho^{w,m_w+i}[(1-\epsilon_i) j + \epsilon_i T_w], \lleft(i), \rright(i)) \,.
\]
In words, if $\epsilon_i=1$, we present the hard sequence on the next copy of ${\cal M}_w$ on the left path, 
and on the right path we ``stay put" (i.e., we keep requesting the target of the last copy of ${\cal M}_w$ that 
was previously traversed). Conversely, if $\epsilon_i=0$ we ``stay put" on the left path, and on the right path 
we traverse the next copy of ${\cal M}_w$ using the hard sequence. (Recall that always $\rho^w[T_w] = \{t\}$.)

\paragraph{Stage 3:} 
Assume that $\lleft(m_{w}) \geq \rright(m_{w})$ (the other case is dealt with by symmetry). Then we set for every 
$i = 1, \hdots, 3 m_{w} - \rright(m_{w})$, and for every $j \in \{1,2,\dots,T_w\}$,
\[
\rho^{w+1}[(2 m_{w} + i - 1) T_w + j] = (\emptyset , \rho^{w,2m_w + i}[j], 1, \rright(m_w) + i) \,.
\]
In words, the path along which we advanced more at the end of stage 2 is ``killed." We continue advancing along
the other path, until we reach the final target (which is the target special vertex in ${\cal M}_{w+1}$).

\subsection{The cost analysis}

Let's scale uniformly the edge weights in ${\cal M}_{w+1}$, so that the diameter of any copy of ${\cal M}_w$
becomes equal to $1$ (this is just for notational simplicity). In particular the optimal cost is $3 m_w$. Let us
assume by induction that we have proved for ${\cal M}_w$ a lower bound of $C_w$ on the competitive ratio
of any deterministic algorithm against the random sequence $\rho^w$. We now analyze the cost of any
deterministic algorithm in the three stages of the random sequence $\rho^{w+1}$.

\paragraph{Stage 1:} 
Here, by induction, the expected cost is simply lower bounded by $m_w C_w$.

\paragraph{Stage 2:} 
Assume that $m_w \geq C_w$. (Note that this is a somewhat stricter assumption than seems to be needed here, but we
will need this stricter condition later.)
Then we claim that the expected cost in this stage is lower bounded by $\frac{m_w}{2} C_w$. Indeed, in each new ``phase" (where the hard sequence $\rho^w$ is presented either on the left path or on the 
right path), the algorithm, with probability $1/2$, has to choose between traversing a copy of ${\cal M}_w$ 
against $\rho^w$, at expected cost $C_w$, or alternatively backtracking (i.e., switching) to the other path (either right away, 
or after a while) which costs at least $2m_w \geq C_w$. This tentatively concludes the proof of the claim, up to the minor issue that
a priori we do not control the expected cost of traversal of ${\cal M}_w$ {\em conditioned on the fact that the algorithm does not switch}. To see why this might be a problem, consider the fictitious situation (indeed, our request sequence will {\em not} be like this) where with probability $1-\epsilon$ the traversal is easy (say cost $0$) and with probability $\epsilon$ the traversal is hard (say cost $1/\epsilon$ times the expectation). In that case switching when the hardness of the traversal is revealed might lower the expected cost by a multiplicative factor $\epsilon$. We explain below in Section~\ref{sec:escape} how to deal with this issue and complete the proof of this paragraph's claim.

\paragraph{Stage 3:} 
This is the key part of the argument. We have $\min\{\lleft(m_w),\rright(m_w)\} \leq \frac{3}{2} m_w - \Omega(\sqrt{m_w})$ with high probability (e.g., on account of the Berry-Esseen Inequality). Therefore, with high probability we face in this stage at least 
$\frac{3 m_w}{2} + \Omega(\sqrt{m_w})$ copies of ${\cal M}_w$ to traverse. Thus, the expected cost of this stage is 
$\left(\frac{3 m_w}{2} + \Omega(\sqrt{m_w}) \right) C_w$.

\subsection{Selecting the parameters and conclusion}

Overall the analysis shows that the algorithm pays in expectation $(3 m_w + \Omega(\sqrt{m_w}) ) C_w$, whereas OPT pays
just $3 m_w$. Thus, we get the recurrence relation $C_{w+1} \geq \left(1 + \Omega\left(\frac{1}{\sqrt{m_w}}\right) \right) C_w$. Clearly,
we want to choose $m_w$ to be as small as possible. With the constraint $m_w \geq C_w$ from the stage 2 analysis we obtain
\[
C_{w+1} \geq C_w + \Omega(\sqrt{C_w}) \,.
\]
In particular we easily get by induction $C_w \geq \Omega(w^2)$. Finally, let us calculate $|M_w|$. We have 
$|M_{w+1}| \le 6 m_w |M_w| \le 6^w \prod_{w'=1}^w m_{w'}$. As $C_w$ is of order $w^2$, and $m_w$ is of order of $C_w$, we get that 
$|M_w|$ is of order $\exp( C w \log w)$ for some constant $C >0$. In particular, if we denote that number 
by $n$, we have $\log n = C w \log w$, so that $w=\Theta\left(\frac{\log n}{\log \log n}\right)$. To extend the
bound to any number of points $n'$, simply choose the largest $w$ for which $|M_w|\le n'$, and then 
extend ${\cal M}_w$ arbitrarily to contain exactly $n'$ points (the extra points will be ignored in the request 
sequence). This concludes the proof of the lower bound up to the conditioning issue in the stage 2 analysis, which we deal with next.

\subsection{Escape price} \label{sec:escape}

Here we resolve the issue of controlling the expected cost obtained by induction for a traversal of ${\cal M}_w$, 
when the algorithm is allowed to abort this traversal and switch to the other branch in ${\cal M}_{w+1}$.
Switching to the other branch in stage $2$, which is the only stage where it is not obvious that aborting 
cannot help, costs at least $2 m_w \geq 2 C_w = 2\diam({\cal M}_w) C_w$. (Recall that when considering
${\cal M}_{w+1}$, we scale the distances so that the diameter of each copy of ${\cal M}_w$ is $1$.)

To control this cost, we actually show inductively a lower bound on the competitive ratio of the escape price 
relaxation of $\MSS$ in ${\cal M}_w$, rather than on $C_{\rand}^{\MSS}({\cal M}_w,|M_w|)$. Let $p_w$ denote the 
escape price for the game in ${\cal M}_w$. We now show by induction that, for $p_w = 2\diam({\cal M}_w) C_w$, 
the escape price option does not enable an online algorithm facing the hard random sequence $\rho^w$ to
overcome the stated lower bound. This in turn will conclude the proof.

The base case is trivial, so let us assume it is true for some $w$. In ${\cal M}_{w+1}$ the escape price is $2\diam({\cal M}_{w+1}) C_{w+1} = 6 m_w C_{w+1}$. (Recall again the uniform scaling of distances when 
considering ${\cal M}_{w+1}$.) Now, assume that the algorithm decides to escape (at the induction level $w+1$), 
and that this escape happens in some copy of ${\cal M}_w$. Using the induction hypothesis, the cost of instead escaping only
at the lower induction level $w$ in the aborted copy, but then proceeding to serve the rest of the hard request sequence 
by going through all the remaining copies of ${\cal M}_w$ and paying the expected cost there, is at most the following:
$2\diam({\cal M}_w) C_w = 2 C_w$ for the escape price, plus $m_w C_w$ for stage $1$, plus $m_w C_w / 2$ for 
stage $2$, plus $2 m_w C_w$ for stage $3$. 
Overall, this cost is at most $(2 + 3.5 m_w) C_w\le 6 m_w C_{w+1}$, as $m_w\ge 1$. In other words, the escape price 
option on ${\cal M}_{w+1}$ does not add any benefit to the algorithm compared to the escape price option on ${\cal M}_w$, 
and thus by induction it does not enable the algorithm to overcome the lower bound.

\subsection{\texorpdfstring{\boldmath Lower bound for $\LGT$}{Lower bound for LGT}}\label{sec:LGT}

We briefly describe how this construction (slightly modified) also gives an $\Omega(w^2)$ lower bound for MSS with sets of size at most $w$ (and hence for layered graph traversal). Notice that in stages 2 and 3 the sets being requested in $\rho^{w+1}$ have the size of the sets in $\rho^w$ plus at most $1$. If this property was also true for stage 1 we would be done. To achieve this, we modify stage 1 as follows: redefine ${\cal M}_{w+1}$ by replacing the first $m_w$ copies of $\cM_w$ on the left path (and similarly on the right path) by a single edge of length $\frac{C_w}{m_w}$ followed by $m_w^2$ copies of $\frac{1}{m_w}\cM_w$, where $\frac{1}{m_w}\cM_w$ denotes the metric space $\cM_w$ with distances scaled by $\frac{1}{m_w}$. The request sequence of stage 1 is similar to that of stage 2, in each step presenting the hard sequence on the next copy of $\frac{1}{m_w}\cM_w$ on one path while staying put on the other path, but rather than choosing the advancing path at random, we alternate between the two sides. In each pair of steps, the expected online cost is at least $\frac{C_w}{m_w}$ (either due to movement through the advancing copy on the path where the online algorithm is located, or to switch to the other path). Since stage 1 consists of $m_w^2$ such pairs of steps, the expected cost of stage 1 is still at least $m_wC_w$, so we get the same lower bound on the online cost as before. The cost of OPT is slightly higher by an additive $\frac{C_w}{m_w}\le 1$ due to the length of the initial extra edge, but since this is much smaller than $m_w$ we still get a recurrence of the same form $C_{w+1} \geq C_w + \Omega(\sqrt{C_w})$, yielding $C_w=\Omega(w^2)$.


\section{\texorpdfstring{\boldmath An Existential $\Omega(\log^2 n)$ Lower Bound for $\MSS$}{An Existential Omega((log n)\textasciicircum 2) Lower Bound for MSS}}\label{sec: main}

In this section, we refine the bound from the previous section and show the following theorem.

\begin{theorem}\label{thm:mainMSS}
	For each $n\in\mathbb N$, there exists an $n$-point metric space $\cM$ such that 
	$C_{\rand}^{\MSS}({\cal M},n-1)=\Omega(\log^2 n)$.\footnote{In fact, our construction of the bad sequence uses only sets of size $O(n^{\log_6 2})$ instead of $n-1$.}
\end{theorem}

By Proposition~\ref{pr: MTS-kSRV-MSS}, this implies that there exists no $o(\log^2n)$-competitive algorithm for $\MTS$ on general $n$-point metrics, and no $o(\log^2k)$-competitive algorithm for $\kSRV$ on general $(k+1)$-point metrics. These lower bounds are tight for $\MTS$, and for $\kSRV$ at least in the case of $(k+1)$-point metrics, on account of the known $O(\log^2n)$ upper bound for $\MTS$ on arbitrary $n$-point metrics~\cite{BCLL19}.

To avoid the $\log^2 \log n$ divisor in the bound proved in Section~\ref{sec: basic}, we will use only $6$ instead of $6m_w$ copies of $\cM_w$ when constructing $\cM_{w+1}$, so that the metric space is of smaller size. Note that the reason why $m_w$ needed to be chosen large in Section~\ref{sec: basic} is to ensure that the cost of switching between the left and right paths is large. The key idea that will enable us to allow smaller switching cost here is that rather than issuing the inductive request sequence in $\cM_w$ in its entirety each time, we can break it into smaller subsequences (``chunks'') and only issue a single chunk at a time. Since the cost of a chunk is smaller than the cost of the entire inductive sequence, we do not need as large of a switching cost to discourage the algorithm from switching between the left and right paths. A similar idea of decomposing inductive request sequences into chunks is also the key idea that will allow us to tighten the universal lower bounds later in Section~\ref{sec: universal}.

Let $0<\alpha<1$ and $\beta\in\mathbb N$ be constants that we determine later. We will 
show that for every $w\in\mathbb N$ there exists a metric space ${\cal M}_w = (M_w,d_w)$ 
with $|M_w|\le \beta\cdot 6^w$ where any randomized algorithm for $\MSS$ (with request sets of size at most $2^w$) has competitive 
ratio at least $\alpha w^2$. The metric space ${\cal M}_w$ has two special points $s_w,t_w\in M_w$ 
that we use as the initial and final location of the lower bound instance (i.e., $s_w$ is the initial 
location of the server, and the request sequence will force the server to terminate at $t_w$). 
The request sequence will be such that an optimal offline algorithm can serve it for cost 
$d_w(s_w,t_w)$ by moving along a shortest path from $s_w$ to $t_w$.

If $\alpha w^2\le 1$, the lower bound is trivial and we choose ${\cal M}_w$ to be $\beta+1$ equally 
spaced points on a line, with $s_w$ and $t_w$ being the two outermost points.

For larger $w$, we proceed by induction. Specifically, for $w\in\mathbb N$ with $\alpha(w+1)^2>1$, 
we construct ${\cal M}_{w+1}$ by taking six copies of ${\cal M}_{w}$ and gluing them together in a
circle as shown in Figure~\ref{fig:constructionStep6}. We view $\cM_{w+1}$ as consisting of a left path and right path of $3$ copies of $\cM_w$ each. For $j\in\{1,2,3\}$, we denote by ${\cal M}_{w,L,j} = (M_{w,L,j},d_w)$ and ${\cal M}_{w,R,j} = (M_{w,R,j},d_w)$ 
the $j$th copy of ${\cal M}_{w}$ on the left and right, respectively. If $p$ is a (set of) point(s) in $M_{w}$, 
we denote by $p_{L,j}$ and $p_{R,j}$ the corresponding (sets of) points in the respective copy. Some members of the family of metric spaces constructed in this way are depicted in Figure~\ref{fig:badSpaceManyLevels}.

\begin{figure}
	\begin{center}
	\includegraphics[width=0.6\textwidth]{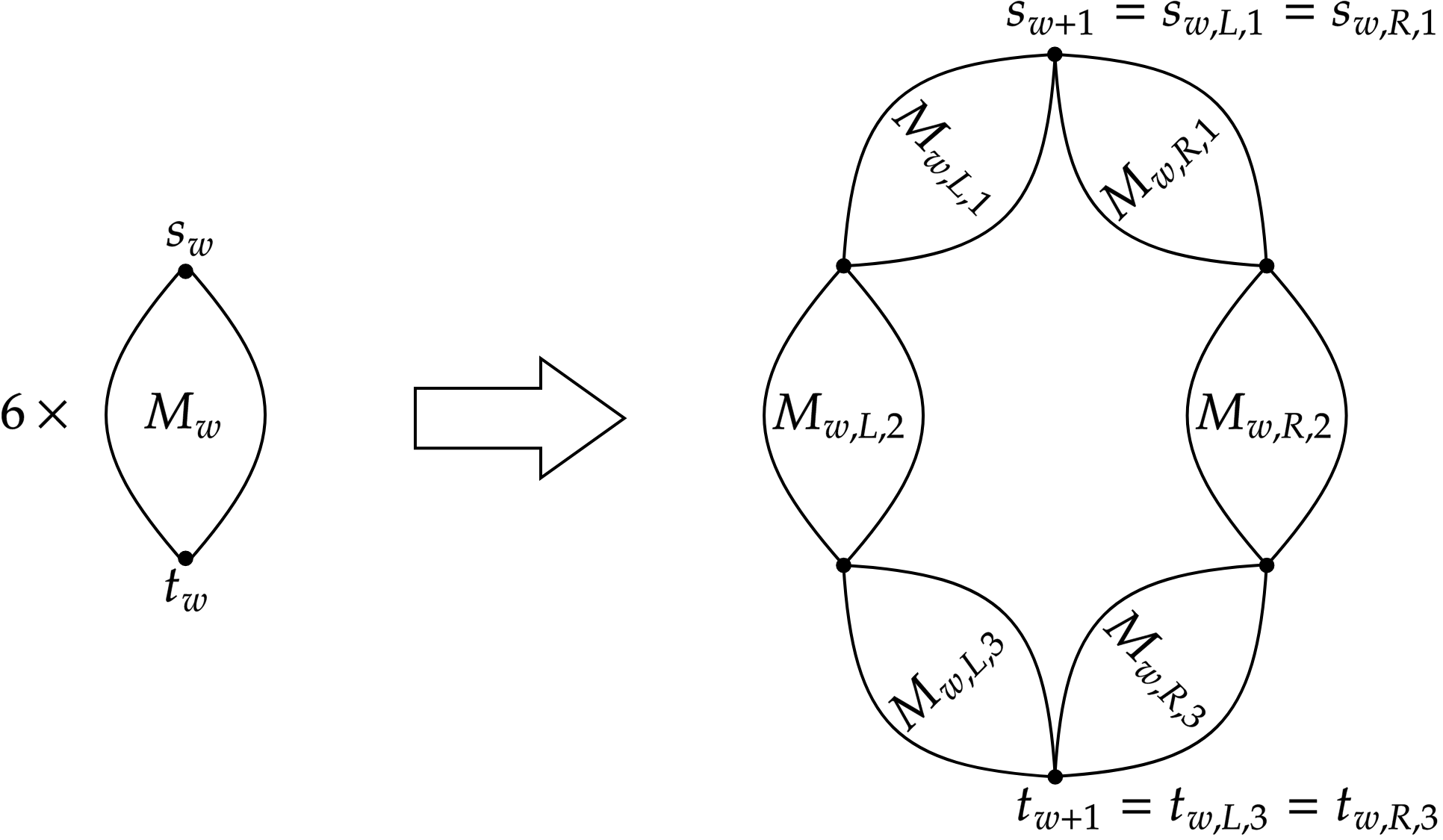}
	\end{center}
	\caption{Construction of $\cM_{w+1}$ and choice of $s_{w+1}$ and $t_{w+1}$ for the $\Omega(\log^2 n)$ lower bound.}\label{fig:constructionStep6}
\end{figure}

\begin{figure}
\begin{center}
	\includegraphics[width=0.99\textwidth]{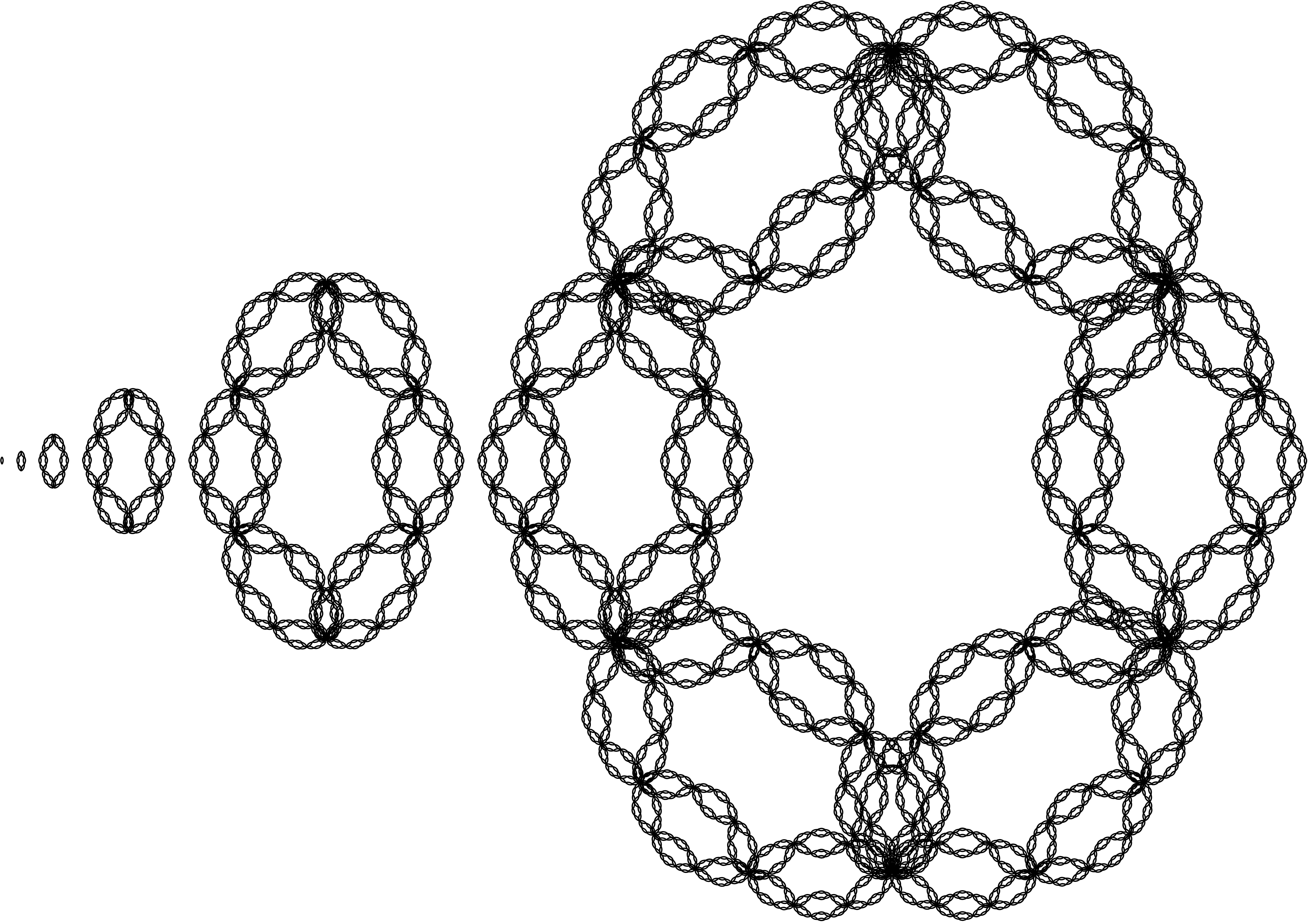}
\end{center}
\caption{Metric spaces where the competitive ratio of MTS and $(n-1)$-server is $\Theta(\log^2 n)$.}\label{fig:badSpaceManyLevels}
\end{figure}

To prove the lower bound, we construct by induction on $w$ a random request sequence $\rho$ for 
${\cal M}_w$ on which \emph{every} deterministic online algorithm has expected cost at least 
$\alpha w^2 d_w(s_w,t_w)$. However, to construct these sequences, we will require an induction 
hypothesis that yields stronger and more delicate properties. Roughly, we will require that the random 
request sequence $\rho$ can be decomposed into  $m\ge \alpha\beta w^2$ (random) 
subsequences $\rho_1,\dots,\rho_{m}$ such that any online algorithm's expected cost on each $\rho_i$ is at least 
some $c_i\approx d_w(s_w,t_w)/\beta$, and the expected sum of these $c_i$ is at least $\alpha w^2d_w(s_w,t_w)$. 
This will be captured more precisely by Lemma~\ref{lem:main} below.

\paragraph{Terminology and notation.} 
A \emph{chunk} is a sequence of requests. We use this term especially for cases where a sequence is 
meant to be used as a subsequence of a longer sequence. Sometimes we also use the term \emph{subchunk} 
instead of chunk, to emphasize that a sequence will be used as a subsequence of a chunk, which in turn 
will be a subsequence of a longer sequence. Recall that in the context of a sequence of chunks 
$\rho_1\rho_2\dots\rho_{m}$, we denote by $\rho_{\le i}=\rho_1\dots\rho_{i}$ the concatenation of the first $i$ 
chunks, and by $\rho_{\le 0}$ an empty sequence. If the $\rho_i$ are random, we denote by $\sigma(\rho_{\le i})$ 
the $\sigma$-algebra generated by $\rho_{\le i}$. The property that $c_i$ is $\sigma(\rho_{\le i-1})$-measurable 
in the following lemma simply means that $c_i$ is a function of $\rho_{\le i-1}$.

\begin{lemma}\label{lem:main}
	For every $w\in\mathbb N$,
	there exists a random sequence of chunks $\rho_1\rho_2\dots\rho_{m}$ for $\MSS$ in 
	${\cal M}_w$, and a choice of random variables $c_1,\dots,c_{m}$, such that for all $i=1,\dots,m$
	the following conditions hold.
	\begin{enumerate}
		\item The last request of $\rho_{m}$ is $\{t_w\}$.
		\item $c_{\opt}(\rho_{\le m})=d_w(s_w,t_w)$.
		\item $c_i$ is $\sigma(\rho_{\le i-1})$-measurable.
		\item For every choice of $\rho_{\le i-1}$ in the sample space,
		         $\E\left[c_{\alg}(\rho_{i}\mid \rho_{\le i-1})\bigm\vert \rho_{\le i-1}\right]\ge c_i$ for any 
		         deterministic online algorithm $\alg$, even if an escape price of $2d_w(s_w,t_w)$ is available 
		         on the suffix $\rho_i$.\label{it:costChunk}
		\item $c_i\in\left[\frac{d_w(s_w,t_w)}{2\beta} ,\frac{3d_w(s_w,t_w)}{2\beta}\right]$.\label{it:ciBound}
		\item $\E\left[\sum_{i=1}^{m} c_i\right] \ge \alpha w^2d_w(s_w,t_w)$.\label{it:sumSizes}
		\item $m\ge \alpha\beta w^2$.\label{it:mBound}
	\end{enumerate}
\end{lemma}

In the statement of Lemma~\ref{lem:main}, we call the number $c_i$ the \emph{size} of the chunk $\rho_i$. 
Intuitively, the size of a chunk captures (a lower bound on) the online cost for serving the chunk as part of the 
longer request sequence $\rho_{\le m}$.

Before proving Lemma~\ref{lem:main}, we briefly argue that it implies Theorem~\ref{thm:mainMSS}.

\begin{proof}[Proof of Theorem~\ref{thm:mainMSS}]
	In the setting of Lemma~\ref{lem:main}, we have
	\begin{align*}
		\E[c_{\alg}(\rho_{\le m})] &= \E\left[\sum_{i=1}^{m} \E\left[c_{\alg}(\rho_{i}\mid \rho_{\le i-1})\bigm\vert \rho_{\le i-1}\right]\right]\\
		&\ge \E\left[\sum_{i=1}^{m} c_i\right]\\
		&\ge \alpha w^2 d(s_w,t_w)\\
		&= \alpha w^2 c_{\opt}(\rho_{\le m}).
	\end{align*}
	Due to the symmetry in the metric space ${\cal M}_w$, we can achieve cost ratio $\alpha w^2$ also in the opposite direction from $t_w$ to $s_w$, so the instance can be  repeated indefinitely\footnote{allowing for an arbitrarily large additive constant $\kappa$ in the definition of competitive ratio}. Since a randomized algorithm is a random distribution over deterministic algorithms, we conclude that for each randomized algorithm there exists a request sequence of arbitrarily high cost where the ratio between expected online cost and optimal cost is at least $\alpha w^2$. Since $\alpha$ and $\beta$ are constants and $n=|M_w|\le \beta \cdot 6^w$, we have $\alpha w^2\in\Omega(w^2)\subseteq\Omega(\log^2 n)$.
\end{proof}

The proof of Lemma~\ref{lem:main} is rather long, so let us first provide a roadmap to its structure. 
We prove Lemma~\ref{lem:main} by induction on $w$. The base case of the induction are all $w$  
with $\alpha w^2\le1$: Here, the metric space is $M_w=\{0,1,\dots,\beta\}$ with the line metric, and 
$s_w=0$ and $t_w=\beta$ so that $d_w(s_w,t_w)=\beta$. Each chunk $\rho_i$ consists only of 
one singleton request. The chunks are sorted from $\{1\}$ to $\{\beta\}$. I.e., $m=\beta$ and 
$\rho_i=\{i\}$ for $i=1,\dots,m$. For $c_i=1$, it is easy to check that all properties are satisfied.

For the induction step, consider some $w+1$ with $\alpha \cdot(w+1)^2>1$. Assuming by the
induction hypothesis that Lemma~\ref{lem:main} holds for $w$, we will show that it also holds 
for $w+1$. This is accomplished by first creating ``subchunks'' that are smaller than required, 
but satisfy similar properties to those in Lemma~\ref{lem:main}, and then combining them to 
chunks of appropriate size that satisfy all the properties.

To simplify notation, assume from now on (without loss of generality) that distances are scaled 
such that $d_w(s_{w},t_{w})=\beta$. Thus, $d_{w+1}(s_{w+1},t_{w+1})=3\beta$. We prove the 
induction step as follows. Rather than constructing the chunks $\rho_1,\dots,\rho_m$ 
and sizes $c_1,\dots,c_m$ for $\cM_{w+1}$ directly, we will first construct subchunks  
$\tilde\rho_1,\dots,\tilde\rho_{\tilde m}$ of smaller sizes $\tilde c_1,\dots,\tilde c_{\tilde m}$ that 
satisfy similar properties. The statement we prove is formalized in Claim~\ref{cl:createSubchunks}. 
Note that the first three properties in the claim are identical to the ones required to complete the induction step of Lemma~\ref{lem:main}, 
but the remaining properties are slightly different: Property~\ref{it:subchunkCost} grants an escape 
price of $2\beta$ instead of $6\beta$, the interval in property~\ref{it:subchunkSize} is $\left[0,\frac{3}{2}\right]$ 
instead of $\left[\frac{3}{2},\frac{9}{2}\right]$, property~\ref{it:subchunkTotalCost} is slightly strengthened 
by a ``$+3$" term, and no lower bound on $\tilde m$ is claimed.
\begin{claim}\label{cl:createSubchunks}
         Let $w\in\mathbb N$ be a natural number for which the properties listed in Lemma~\ref{lem:main} hold,
         and assume that $\alpha \cdot(w+1)^2>1$.
	 There exists a random sequence of subchunks $\tilde\rho_1,\tilde\rho_2,\dots,\tilde\rho_{\tilde m}$ 
	 in ${\cal M}_{w+1}$ and a choice of random variables $\tilde c_1,\dots,\tilde c_{\tilde m}$ such that 
	 for all $i=1,\dots,\tilde m$ the following conditions hold.
	\begin{enumerate}
		\item The last request of $\tilde\rho_{\tilde m}$ is $\{t_{w+1}\}$.
		\item $c_{\opt}(\tilde\rho_{\le \tilde m})=3\beta$.
		\item $\tilde c_i$ is $\sigma(\tilde\rho_{\le i-1})$-measurable.
		\item For every choice of $\tilde \rho_{\le i-1}$ in the sample space, 
		$\E\left[c_{\alg}(\tilde\rho_{i}\mid \tilde \rho_{\le i-1})\bigm\vert \tilde\rho_{\le i-1}\right]\ge \tilde c_i$ 
		for any deterministic online algorithm $\alg$, even if an escape price of $2\beta$ is available on the 
		suffix $\tilde\rho_i$.\label{it:subchunkCost}
		\item $\tilde c_i\in\left[0,\frac{3}{2}\right]$.\label{it:subchunkSize}
		\item $\E\left[\sum_{i=1}^{\tilde m} \tilde c_i\right] \ge \alpha (w+1)^2 \cdot3\beta \,\,\,+ \,\,3$.\label{it:subchunkTotalCost}
	\end{enumerate}
\end{claim}

Thus, the proof of Lemma~\ref{lem:main} by induction proceeds in two steps. We first prove 
Claim~\ref{cl:createSubchunks}. This asserts that if the properties listed in Lemma~\ref{lem:main} 
hold for $w$, then some weaker properties hold for $w+1$. Then we prove that if 
Claim~\ref{cl:createSubchunks} holds, then also the properties listed in Lemma~\ref{lem:main} 
hold for $w+1$. This is done in the following sections. In Section~\ref{sec: subchunks} we construct 
the subchunks and sizes stipulated in Claim~\ref{cl:createSubchunks}. In Section~\ref{sec:subchunkAna} 
we prove the weak inductive step stated in Claim~\ref{cl:createSubchunks}. Finally, in 
Section~\ref{sec:combiningSubchunks} we deduce the properties in Lemma~\ref{lem:main} 
for $w+1$ from Claim~\ref{cl:createSubchunks}.

\subsection{Constructing subchunks}\label{sec: subchunks}

Suppose Lemma~\ref{lem:main} holds for some fixed $w$ with $\alpha \cdot(w+1)^2>1$. To deduce 
the properties in Claim~\ref{cl:createSubchunks} from the induction
hypothesis, we first describe in this subsection the construction of the subchunks 
$\tilde\rho_1,\tilde\rho_2,\dots,\tilde\rho_{\tilde m}$ and sizes $\tilde c_1,\dots,\tilde c_{\tilde m}$.
Recall that distances are scaled such that $d_w(s_{w},t_{w})=\beta$.

The 
request sequence $\tilde\rho_{\le\tilde m}$ consists of three stages that are executed one after another. 
In stage $j$, requests will only involve points in $M_{w,L,j}\cup M_{w,R,j}$.

\paragraph{Stage 1:} 
Sample a sequence of chunks in ${\cal M}_w$ by the induction hypothesis, and replace 
each request set $u\subseteq {\cal M}_w$ in it by $u_{L,1}\cup u_{R,1}\subseteq M_{w,L,1}\cup M_{w,R,1}$. 
The sequence obtained this way constitutes the start of our sequence of subchunks $\tilde\rho_1,\tilde\rho_2,\dots$, 
and we define the size $\tilde c_i$ of each subchunks of stage 1 to be equal to the size of the corresponding chunk 
from which it is created.

\paragraph{Stage 2:} 
Invoking the induction hypothesis two more times,\footnote{Recall that we are assuming inductively the stronger
statement of Lemma~\ref{lem:main}} sample independently two sequences of chunks $\rho_{L,1},\dots,\rho_{L,m_L}$ 
and $\rho_{R,1},\dots,\rho_{R,m_R}$ in $M_{w,L,2}$ and $M_{w,R,2}$, respectively, and let 
$c_{L,i}, c_{R,i}\in\left[\frac{1}{2},\frac{3}{2}\right]$ be the corresponding sizes.

To construct a subchunk in stage 2, we repeatedly do the following: Select a chunk from one of the two inductively 
constructed sequences of chunks, say $\rho_{L,i}$, and replace in it each request set by the union of itself and the 
set $u$ of points on the \emph{other} side (i.e., on the right when the selected chunk is $\rho_{L,i}$) that were 
contained in the last request set of the previous subchunk. We then say that the construction of this subchunk 
\emph{uses} chunk $\rho_{L,i}$. The chunk used to construct a subchunk in this way is always the first unused 
chunk from either the left or the right (with the side determined randomly in a way described shortly). Thus, each 
subchunk advances the request sequence on one of the two sides by one chunk, incurring cost for any algorithm 
that has its server on that side, whereas an algorithm would pay cost $0$ for the subchunk if it has its server on 
the other side (since it would already be at a point of $u$).

Let $\lleft(j)$ and $\rright(j)$ be the number of chunks used on the left and right for the construction 
of the first $j$ chunks of stage $2$ (so $\lleft(j)+\rright(j)=j$). Let 
\begin{align*}
	n_{L,j}:=c_{L,\lleft(j-1)+1} \qquad\text{ and } \qquad n_{R,j}:=c_{R,\rright(j-1)+1} 
\end{align*}
be the sizes of the \emph{next unused} chunks on the two sides immediately \emph{before} the construction 
of subchunk $j$ of stage $2$. These are the two candidate chunks for use in the construction of subchunk $j$. 
Then the construction of the $j$th subchunk of stage 2 uses the next unused chunk on the left with probability 
$p_{L,j}$, and it uses the next unused chunk on the right otherwise (with probability $p_{R,j} = 1 - p_{L,j}$), where
\begin{align*}
	p_{L,j}:=\frac{n_{R,j}}{n_{L,j}+n_{R,j}} \qquad\text{ and } \qquad p_{R,j}:=\frac{n_{L,j}}{n_{L,j}+n_{R,j}}.
\end{align*}
These probabilities are chosen so that if we let $L_j:=\sum_{i=1}^{\lleft(j)} c_{L,i}$ and 
$R_j:=\sum_{i=1}^{\rright(j)} c_{R,i}$ be the total size of chunks used on the left and right, respectively, for the first 
$j$ subchunks of stage 2, then for $S_j:=L_j-R_j$, the sequence $(S_j)_{j=1,2,\dots}$ is a martingale. We set the size of the 
$j$th subchunk to be $n_{L,j} n_{R,j} / (n_{L,j} + n_{R,j})$.

These steps continue until the number of subchunks created in stage 2 is
\begin{align}
	\kappa=\min\left\{k\colon \sum_{j=1}^{k+1} n_{L,j} n_{R,j} \ge \frac{\alpha \beta w^2}{4}\right\}.\label{eq:defk}
\end{align}
Note that it can be checked after the construction of each subchunk whether this stopping condition is satisfied, as 
it only depends on random choices made beforehand. Since each summand $n_{L,j} n_{R,j} \ge \frac{1}{4}$ (as the
inductive hypothesis assumes Property~\ref{it:ciBound} of Lemma~\ref{lem:main}), we have 
$\kappa< \alpha\beta w^2\le \min\{m_L,m_R\}$. Thus, the stopping condition is reached before running out of chunks.

We refer to these first $\kappa$ subchunks of stage 2 as \emph{stage 2a}, and the following part as \emph{stage 2b}: 
The subchunks of stage 2b are simply the unused chunks on the side whose total size of used chunks in stage 2a 
was smaller, and the other side gets ``killed''. More precisely, if $L_\kappa\le R_\kappa$ then the subchunks of stage 2b are 
$\rho_{L,\lleft(\kappa)+1}, \rho_{L,\lleft(\kappa)+2}, \dots, \rho_{L,m_L}$, and otherwise they are 
$\rho_{R,\rright(\kappa)+1}, \rho_{R,\rright(\kappa)+2}, \dots, \rho_{R,m_R}$. The corresponding sizes are the ones 
given by the induction hypothesis.

\paragraph{Stage 3:} The last request of stage 2b was either $\{t_{w,L,2}\}=\{s_{w,L,3}\}$ or $\{t_{w,R,2}\}=\{s_{w,R,3}\}$. 
In the former case, the subchunks and sizes of stage $3$ are obtained by invoking the induction hypothesis in 
${\cal M}_{w,L,3}$, and otherwise in ${\cal M}_{w,R,3}$.

\subsection{Analysis}\label{sec:subchunkAna}

Let $\tilde\rho_1,\dots,\tilde\rho_{\tilde m}$ and $\tilde c_1,\dots,\tilde c_{\tilde m}$ be the entire sequences of subchunks and sizes constructed in stages 1, 2 and 3. The first three properties in Claim~\ref{cl:createSubchunks} follow immediately from the construction and the induction hypothesis of Lemma~\ref{lem:main}. In particular, an optimal offline algorithm can serve the request sequence for cost $3\beta$ by moving through the three copies of ${\cal M}_w$ that are on the side (left or right) where stages 2b and 3 are played.

Property~\ref{it:subchunkCost} follows easily from the induction hypothesis if $\tilde \rho_i$ belongs to stage 1, 2b or 3. If $\tilde\rho_i$ belongs to stage 2a, say it is the $j$th subchunk of stage 2a, then suppose the algorithm is on the left side right before this subchunk is issued. (The case that the algorithm is on the right is symmetric.) With probability $p_{L,j}$ the subchunk is constructed using the next unused inductive chunk on the left. Conditioned on this being the case, the algorithm pays expected cost at least $n_{L,j}$ for this subchunk by the induction hypothesis. Indeed, even if the algorithm switches to a point in $u$ (the set on the right included in each request) in order to avoid paying more cost for the subchunk, switching to $u$ costs at least $2\beta$, which is the escape price granted by the induction hypothesis. As the event we conditioned on has probability $p_{L,j}$, the (unconditioned) expected cost of the algorithm for the subchunk is at least $p_{L,j} n_{L,j}=n_{L,j}n_{R,j}/(n_{L,j}+n_{R,j})$. This is precisely the size of this subchunk by definition, proving property~\ref{it:subchunkCost}.

The property $\tilde c_i\in\left[0,\frac{3}{2}\right]$ is immediate from the induction hypothesis for stages 1, 2b and 3, recalling $\beta=d_w(s_w,t_w)$. For stage 2a it also holds because the size of the $j$th subchunk is $n_{L,j}n_{R,j}/(n_{L,j}+n_{R,j})\le \max\{n_{L,j}, n_{R,j}\} \le \frac{3}{2}$.

Finally and most crucially, we wish to show that
\begin{align}
	\E\left[\sum_{i=1}^{\tilde m} \tilde c_i\right] \ge 3\alpha\beta (w+1)^2 + 3.\label{eq:goalBound}
\end{align}
By construction and the induction hypothesis, the expected subchunk sizes sum to at least $\alpha\beta w^2$ in stage 1, $\alpha\beta w^2-\E\left[\min\{L_\kappa,R_\kappa\}\right]$ in stage 2b, and $\alpha\beta w^2$ in stage 3. To analyze stage 2a, recall the definitions $L_j:=\sum_{i=1}^{\lleft(j)} c_{L,i}$ and $R_j:=\sum_{i=1}^{\rright(j)} c_{R,i}$ and $S_j:= L_j-R_j$.
Notice that the size of the $j$th subchunk of stage 2a, which is $n_{L,j} n_{R,j}/(n_{L,j} + n_{R,j})=p_{L,j} n_{L,j}$, is precisely equal to the expectation of $L_j-L_{j-1}$ (conditioned on prior randomness). Thus, the expected sum of chunk sizes in stage 2a is $\E[L_\kappa]$, which by symmetry is equal to $\E[R_\kappa]$ and thus equal to $\frac{1}{2}\E[L_\kappa+R_\kappa]$. Combining these lower bounds for all stages, we obtain

\begin{align}
	\E\left[\sum_{i=1}^{\tilde m} \tilde c_i\right] &\ge \underbrace{\alpha\beta w^2}_{\text{stage 1}} + \underbrace{\alpha\beta w^2- \E\left[\min\{L_\kappa,R_\kappa\}\right]}_{\text{stage 2b}} + \underbrace{\alpha\beta w^2}_{\text{stage 3}} + \underbrace{\frac{1}{2}\E[L_\kappa+R_\kappa]}_{\text{stage 2a}}\notag\\
	&= 3\alpha\beta w^2 + \frac{1}{2}\E[|S_\kappa|].\label{eq:boundWithSk}
\end{align}
where the equation uses that $\min\{x,y\}=\frac{x+y-|x-y|}{2}$ for $x,y\in\R$. It remains to bound $\E[|S_\kappa|]$. To do so, we will use the following result about the convergence rate in the martingale central limit theorem.

\begin{lemma}[Ibragimov~\cite{ibragimov}]\label{lem:martingaleBerryEsseen}
	Let $\gamma, \eta>0$ be constants. Let $S_0=0, S_1,S_2,\dots$ be a martingale and let $X_j=S_j-S_{j-1}$. If $|X_j|\le \gamma$ and $\sum_{j=1}^\infty\E[X_j^2|X_1,\dots,X_{j-1}]=\infty$ a.s., then for all $z\in\R$ and $\kappa=\min\left\{k\colon \sum_{j=1}^{k+1}\E[X_j^2|X_1,\dots,X_{j-1}]\ge \eta^2\right\}$ it holds that
	\begin{align*}
		\left|\Pr\left[S_\kappa<z\eta\right]-\Phi(z)\right|\le 2\sqrt{\frac{\gamma}{\eta}}\left(1+\frac{3}{2}\frac{\gamma}{\eta} + \frac{\gamma^2}{3\eta^2}\right),
	\end{align*}
	where $\Phi$ is the standard normal distribution function.
\end{lemma}

Let us first calculate the conditional variances of the martingale difference terms $X_j=S_j-S_{j-1}$ in our setting:\footnote{The first equation in~\eqref{eq:condVar} only makes sense if we consider the two sequences of chunks sampled in stage 2 as fixed (rather than random), so that the sizes $c_{L,i}$ and $c_{R,i}$ may be viewed as constants rather than random variables. Otherwise, $n_{L,j}$ and $n_{R,j}$ would not be measurable with respect to the $\sigma$-algebra generated by $X_1,\dots,X_{j-1}$, which the expectation is conditioned on. Thus, the bound on $\E[|S_k|]$ that we will obtain in \eqref{eq:absSkBound} holds for any fixed choice of the two chunk sequences sampled in stage 2. Therefore, it also holds when the expectation is taken over the randomness of these chunks as well.}
\begin{align}
	\E[X_j^2\mid  X_1,\dots,X_{j-1}]&=\frac{n_{R,j} n_{L,j}^2}{n_{L,j}+n_{R,j}} + \frac{n_{L,j} n_{R,j}^2}{n_{L,j}+n_{R,j}} = n_{L,j} n_{R,j} \in \left[\frac{1}{4},\frac{9}{4}\right]\label{eq:condVar}
\end{align}

The premise $\sum_{j=1}^\infty\E[X_j^2|X_1,\dots,X_{j-1}]=\infty$ in Lemma~\ref{lem:martingaleBerryEsseen} can easily be satisfied in our case by extending the martingale with additional terms after step $\kappa$.\footnote{In the lemma, this condition merely serves to guarantee the existence of $\kappa$ for any $\eta$. For the specific $\eta$ that we will use, we justified the existence of $\kappa$ already when we constructed stage 2a.} We invoke Lemma~\ref{lem:martingaleBerryEsseen} with $\gamma=\frac{3}{2}$ and $\eta=\frac{\sqrt{\alpha\beta}\cdot w}{2}$, so that $|X_j|\le\max\{n_{L,j},n_{R,j}\}\le\gamma$ and the definition of $\kappa$ in the lemma coincides with the one given in~\eqref{eq:defk}. Since $\sqrt{\alpha} \cdot (w+1) > 1$ by the assumption of the induction step, we then have $\frac{\gamma}{\eta} < \frac{3(w+1)}{w\sqrt{\beta}}= O\left(\frac{1}{\sqrt\beta}\right)$. Thus,
\begin{align}
	\E[|S_\kappa|]&\ge \eta\cdot\left(\Pr[S_\kappa<-\eta] + \Pr[S_\kappa \ge \eta]\right)\notag\\
	&\ge \eta\cdot \left(2\Phi(-1)-O\left(\beta^{-1/4}\right)\right)\notag\\
	&\ge \frac{\sqrt{\alpha\beta}\cdot w}{2}\cdot\Phi(-1)\label{eq:absSkBound}
\end{align}
for a sufficiently large constant $\beta\in\N$.

Combining \eqref{eq:boundWithSk} and \eqref{eq:absSkBound} and choosing $\alpha$ such that $\frac{\Phi(-1)}{4}=9\sqrt{\alpha\beta}$, we obtain
\begin{align*}
	\E\left[\sum_{i=1}^{\tilde m} \tilde c_i\right] &\ge 3\alpha\beta w^2 + 9\alpha\beta w\\
	&= 3\alpha\beta (w+1)^2 +3\alpha\beta(w-1)\\
	&> 3\alpha\beta (w+1)^2 +\sqrt{\alpha}\beta\\
	&=3\alpha\beta (w+1)^2 +\frac{\Phi(-1)\sqrt\beta}{36}\\
	&\ge 3\alpha\beta (w+1)^2 + 3
\end{align*}
where the strict inequality uses $\sqrt\alpha\cdot(w+1)>1$ and $3(w-1)\ge w+1$, and the last inequality holds for a sufficiently large constant $\beta$. This completes the proof of Claim~\ref{cl:createSubchunks}.

\subsection{Combining subchunks}\label{sec:combiningSubchunks}

The next lemma shows that random subchunks $\tilde\rho_1,\dots,\tilde\rho_{\tilde m}$ of sizes $\tilde c_i\le \tilde c_{\max}$ can be combined into chunks $\rho_1,\dots,\rho_m$ of greater sizes $c_i\approx\cavg$, and so that $m=\left\lfloor\E\left[\sum_{i=1}^{\tilde m}\tilde c_i\right]/\cavg\right\rfloor$. Note that the lemma is not specific to MSS, so it may also be used for constructing similar chunk decompositions for other problems with recursively defined lower bounds that might benefit from this technique.
\begin{lemma}\label{lem:combiningSubchunks}
	Let $0<\tilde c_{\max}\le \cavg$ be constants. Let $\tilde\rho_1,\tilde\rho_2,\dots,\tilde\rho_{\tilde m}$ be a random sequence of (sub)chunks and let $\tilde c_1,\dots,\tilde c_{\tilde m}\in[0,\tilde c_{\max}]$ be random variables such that $\tilde c_j$ is $\sigma(\tilde\rho_{\le j-1})$-measurable and
	\begin{align}
		\E\left[c_{\alg}(\tilde \rho_{j}\mid\tilde \rho_{\le j-1})\bigm\vert \tilde \rho_{\le j-1}\right]&\ge \tilde c_j\label{eq:combinePremise}
	\end{align}
	for any deterministic online algorithm $\alg$.
	Then there exists a random sequence $\rho_1,\dots,\rho_m$ of chunks and random variables $c_1,\dots,c_m$ such that:
	\begin{enumerate}
		\item $\rho_{\le m}=\tilde\rho_{\le \tilde m}$,
		\item $c_i$ is $\sigma(\rho_{\le i-1})$-measurable,
		\item $\E\left[c_{\alg}(\rho_i\mid\rho_{\le i-1})\bigm\vert \rho_{\le i-1}\right]\ge c_i$ for any deterministic online algorithm $\alg$,\label{it:combineECost}
		\item $c_i\in[\cavg-\tilde c_{\max},\cavg+\tilde c_{\max}]$,\label{it:combinecBound}
		\item $\E\left[\sum_{i=1}^m c_i\right] \ge \E\left[\sum_{i=1}^{\tilde m}\tilde c_i\right]-\cavg$,\label{it:combineSumc}
		\item $m=\left\lfloor\E\left[\sum_{i=1}^{\tilde m}\tilde c_i\right]/\cavg\right\rfloor$.\label{it:combineM}
	\end{enumerate}
	If~\eqref{eq:combinePremise} holds also when an escape price $\pesc$ is available on the suffix $\tilde \rho_j$, then property~\ref{it:combineECost} holds also when an escape price of at least $\pesc+\cavg+\tilde c_{\max}$ is available on the suffix $\rho_i$.
\end{lemma}
The proof of the induction step of Lemma~\ref{lem:main} is completed by applying Lemma~\ref{lem:combiningSubchunks} with $\tilde c_{\max}=3/2$, $\cavg=3$ and $\pesc=2\beta$ to the subchunks guaranteed by Claim~\ref{cl:createSubchunks}. For $\beta\ge 2$, the escape price $6\beta$ in $\cM_{w+1}$ is indeed at least $\pesc+\cavg+\tilde c_{\max}$.
\begin{proof}[Proof of Lemma~\ref{lem:combiningSubchunks}]
Let $C:=\E\left[\sum_{i=1}^{\tilde m}\tilde c_i\right]$ and $m:=\left\lfloor C/\cavg\right\rfloor$, satisfying 
property~\ref{it:combineM}. We define the chunks $\rho_i$ and sizes $c_i$ based on random indices 
$h_i$ as follows. Set for every $i=0,1,\dots, m$, 
$$
h_i := \min\left\{h\colon \E\left[\sum_{j=h+1}^{\tilde m}\tilde c_j\Biggm\vert \rho_{\le h}\right]\le C-i\cdot\cavg\right\}.
$$
Next, put
$$
\rho_i := \left\{\begin{array}{ll}
		        \tilde\rho_{h_{i-1}+1\le\cdot\le h_i} &\hbox{if } i\in\{1,\dots,m-1\},\\
		        \tilde\rho_{h_{m-1}+1\le\cdot\le \tilde m} & \hbox{if } i = m.
		 \end{array}\right.
$$
Finally, put for all $i=1,2,\dots,m$,
$$
c_i := \E\left[\sum_{j=h_{i-1}+1}^{h_i}\tilde c_j\Biggm\vert\tilde\rho_{\le h_{i-1}}\right].
$$
Note that $h_0=0\le h_1\le\dots\le h_m\le \tilde m$, and therefore $\rho_{\le i}= \tilde\rho_{\le h_{i}}$ 
for $i\le m-1$, and $\rho_{\le m}=\tilde\rho_{\le\tilde m}$. In particular, $c_i$ is $\sigma(\rho_{\le i-1})$-measurable 
by design. For any online algorithm $\alg$,
	\begin{align}
		\E\left[c_{\alg}(\rho_i\mid \rho_{\le i-1})\bigm\vert \rho_{\le i-1}\right] &\ge \E\left[\sum_{j=h_{i-1}+1}^{h_i}\E\left[c_{\alg}(\tilde \rho_{j}\mid \tilde \rho_{\le j-1})\bigm\vert \tilde \rho_{\le j-1}\right]\Biggm\vert\tilde\rho_{\le h_{i-1}}\right]\notag\\
		&\ge  \E\left[\sum_{j=h_{i-1}+1}^{h_i}\tilde c_j\Biggm\vert\tilde\rho_{\le h_{i-1}}\right]\label{eq:combineDelicate}\\
		&= c_i.\notag
	\end{align}
	In the setting with escape prices this is still true, but inequality~\eqref{eq:combineDelicate} requires a charging argument: Recall that we assume here that inequality~\eqref{eq:combinePremise} holds for algorithms that are allowed to escape during $\tilde\rho_j$ (but fails for algorithms that escape earlier).
	If $A$ escapes for cost $\pesc+\cavg+\tilde c_{\max}$ during $\rho_i$, let $j^*\in\{h_{i-1}+1,\dots,h_i\}$ be the index of the subchunk $\tilde\rho_{j^*}$ of $\rho_i$ during which it escapes. We charge only $\pesc$ of this cost $\pesc+\cavg+\tilde c_{\max}$ to the subchunk $\tilde\rho_{j^*}$, and we charge a ``fake cost'' of $\E[\tilde c_{j}\mid\tilde\rho_{\le h_{i-1}}]$ to each subsequent subchunk $\tilde\rho_j$ for $j=j^*+1,\dots,h_i$ (where the true cost of $A$ would be $0$ because the algorithm already escaped). So the total amount charged in this way is
	\begin{align*}
		\pesc+\sum_{j=j^*+1}^{h_i}\E[\tilde c_j\mid\tilde\rho_{\le h_{i-1}}]\le \pesc+c_i,
	\end{align*}
	which does not overcharge the escape price of $\pesc+\cavg+\tilde c_{\max}$ provided that the bound $c_i\le \cavg+\tilde c_{\max}$ of property~\ref{it:combinecBound} holds. For this charging scheme, inequality~\eqref{eq:combineDelicate} follows by invoking~\eqref{eq:combinePremise} for $j\le j^*$ and substituting the fake charge $\E[\tilde c_{j}\mid\tilde\rho_{\le h_{i-1}}]$ for $j>j^*$.
	
	To see that property~\ref{it:combinecBound} holds, note that
	\begin{align*}
		c_i&=\E\left[\sum_{j=h_{i-1}+1}^{h_i}\tilde c_j\Biggm\vert\tilde\rho_{\le h_{i-1}}\right]\\
		&=\E\left[\sum_{j=h_{i-1}+1}^{\tilde m}\tilde c_j \Biggm\vert\tilde\rho_{\le h_{i-1}}\right] - \E\left[\E\Biggl[\sum_{j=h_{i}+1}^{\tilde m}\tilde c_j \biggm\vert\tilde\rho_{\le h_{i}}\Biggr]\Biggm\vert\tilde\rho_{\le h_{i-1}}\right].
	\end{align*}
	By definition of $h_{i}$ and the assumption that $\tilde c_j\in[0,\tilde c_{\max}]$, the first term lies in $[C-(i-1)\cdot\cavg-\tilde c_{\max}, C-(i-1)\cdot\cavg]$ and the second in $[C-i\cdot\cavg-\tilde c_{\max}, C-i\cdot\cavg]$. Thus, their difference must lie in $[\cavg-\tilde c_{\max},\cavg+\tilde c_{\max}]$.
	
	Finally, property~\ref{it:combineSumc} holds because
	\begin{align*}
		\E\left[\sum_{i=1}^m c_i\right]&=\E\left[\sum_{i=1}^{h_m} \tilde c_i\right]\\
		&\ge \E\left[\sum_{i=1}^{\tilde m} \tilde c_i\right]-(C-m\cdot\cavg)\\
		&\ge \E\left[\sum_{i=1}^{\tilde m} \tilde c_i\right]-\cavg,
	\end{align*}
	completing the proof.
\end{proof}

\section{\texorpdfstring{A Universal \boldmath$\Omega(\log n)$ Lower Bound for $\MSS$}{A Universal Omega(log n) Lower Bound for MSS}}\label{sec: universal}

We now prove the tight universal lower bounds for $\kSRV$ and $\MTS$. The main technical 
contribution of this section is a lower bound on $\MSS$ in
hierarchically separated tree metrics (HSTs). For definitions and further discussion,
see Appendix~\ref{sec: metrics}. The role that HSTs play in proving universal lower 
bounds is a consequence of the following theorem. 
\begin{theorem}[{Bartal, Linial, Mendel, Naor~\cite[Theorem 3.26]{BLMN03}}]
For every $\epsilon\in (0,1)$ there exists $\delta > 0$ such that for every 
$q\ge 1$, for every $n\in \mathbb N$, and for every $n$-point metric 
space, there exists an $n^{\delta/\log 2q}$-point subspace that is a $q$-HST 
up to bi-Lipschitz distortion $\le 2+\epsilon$.
\end{theorem}

\begin{corollary}\label{cor: HST sufficient}
Let $q:\mathbb N\rightarrow [1,\infty)$. Suppose that for all $n\in \mathbb N$,
that for all $n$-point $q(n)$-HSTs ${\cal H}$ we have that 
$C_{\rand}^{\MSS}({\cal H},n) = \Omega(\log n)$. Then, for all $n\in \mathbb N$,
for all $n$-point metric spaces ${\cal M}$,
$C_{\rand}^{\MSS}({\cal M},n-1) = \Omega(\log n/\log 2q(n))$.
\end{corollary}
The previously best universal lower bound is implied by the $\Omega(\log n)$ 
lower bound of~\cite[Theorem 3]{BBM01} for all $n$-point $\Omega(\log^2 n)$-HSTs. 
Here we improve this lower bound by giving a lower bound of $\Omega(\log n)$ 
for all $n$-point $1$-HSTs. The reason why the lower bound construction in~\cite[Theorem 3]{BBM01} works for $\Omega(\log^2n)$-HSTs but not for $1$-HSTs is that a stretch of $\Omega(\log^2n)$ is needed to ensure that the cost of switching between subtrees is large relative to the cost of an inductive request sequence within a subtree. To overcome this, we use a similar idea to the one in the previous section of decomposing the recursive request sequence into smaller chunks. Apart from this, our construction is similar to the one in \cite{BBM01}.

\begin{lemma}\label{lem:universalMain}
There is a constant $\alpha>0$ such that for any $1$-HST (a.k.a. an ultrametric space) 
${\cal U}' = (U',d)$ there exists a subspace ${\cal U} = (U,d_{|U})$ of ${\cal U}'$ with 
$\diam({\cal U})=\diam({\cal U}')$ and there exists a distribution $\cD$ of request sequences 
in $U$ satisfying the following properties for every initial location $s\in U$:
\begin{itemize}
	\item For any deterministic online algorithm $\alg$, and even if an escape price of 
	         $2\diam({\cal U})$ is available:
	\begin{align}
		\E_{\rho\sim\cD}\left[c_{\alg,s}(\rho)\right]\ge \diam({\cal U})\ge 
		          \E_{\rho\sim\cD}\left[c_{\opt,s}(\rho)\right]. \label{eq:universalCost}
	\end{align}
	\item There exists some $h({\cal U})\in\N$ such that for all $h\ge h({\cal U})$:
	\begin{align}
		h\cdot\diam({\cal U}) \ge \alpha\cdot\log |U'|\cdot \E_{(\rho_1,\dots,\rho_h)\sim\cD^h}
	               \left[c_{\opt,s}(\rho_1\rho_2\dots\rho_h)\right]. \label{eq:universalCR}
	\end{align}
\end{itemize}
\end{lemma}

Before delving into the proof of Lemma~\ref{lem:universalMain}, we state and prove its
consequences, the main results of this section.
\begin{theorem}\label{thm: main universal}
For all $n\in\N$, in \emph{every} $n$-point metric space ${\cal M}$, $C_{\rand}^{\MSS}({\cal M},n-1) = \Omega(\log n)$.
\end{theorem}

\begin{proof}
By Corollary~\ref{cor: HST sufficient}, it is sufficient to prove lower bounds for all $1$-HSTs.
Given an $n$-point $1$-HST ${\cal U}' = (U',d)$, Lemma~\ref{lem:universalMain} asserts that
there exists $h\in\N$ and a probability distribution $\tilde{\rho}$ on request sequences (namely 
$\rho_1\rho_2\dots\rho_h$ in Inequality~\ref{eq:universalCR}), such that for every deterministic
algorithm $\alg$,
$$
\E[c_{\alg}(\rho):\ \rho\sim\tilde{\rho}]\ge h\cdot\diam({\cal U}')\ge 
       \alpha\cdot\log n\cdot \E[c_{\opt}(\rho):\ \rho\sim\tilde{\rho}],
$$
where $\alpha > 0$ is an absolute constant. Therefore,
$C_{\distr}^{\MSS}({\cal U}',n) = \Omega(\log n)$, and the proof is concluded by 
Theorem~\ref{thm: Yao minimax}.
\end{proof}
\begin{corollary}
	For all $k<n\in\mathbb N$, in \emph{every} $n$-point metric space $\cM$, $C_{\rand}^{\MTS}(\cM)=\Omega(\log n)$ and $C_{\rand}^{\kSRV}(\cM)=\Omega(\log k)$
\end{corollary}
\begin{proof}
	From Proposition~\ref{pr: MTS-kSRV-MSS}, noting that the $k$-server problem in $\cM$ is at least as hard as it is in any subspace of size $k+1$.
\end{proof}
These universal lower bounds are asymptotically tight due to the matching upper bounds known in some special metrics.

We may assume without loss of generality that all the internal node weights in the HST representation
of ${\cal U}'$ are at least $1$ (otherwise, just scale all the weights uniformly). By first rounding weights 
to powers of $2$ and then contracting all edges whose incident vertices have the same weight, we may further 
assume without loss of generality (losing a factor $2$ in the constant $\alpha$) that the ratio between 
the weights of two adjacent internal nodes is always of the form $2^i$ for some integer $i\ge 1$. (In
particular, the modified HST is a $2$-HST.)

We prove Lemma~\ref{lem:universalMain} by induction on $n:=|U'|$. The constant $\alpha$ will be 
determined later. For $n=1$ the lemma is trivial. Suppose now that $n\ge2$. We prove the inductive
step in the following subsections. In Section~\ref{sec: universal construction} we construct inductively
the subspace ${\cal U}$ and the probability distribution $\cD$. In Section~\ref{sec: universal cost} we 
prove Inequality~\eqref{eq:universalCost}. In Sections~\ref{sec: universal uniform},~\ref{sec: universal balanced}, 
and~\ref{sec: universal binary} we prove Inequality~\eqref{eq:universalCR}, through a case analysis.

\subsection{Construction}\label{sec: universal construction}

Let ${\cal U}'_1,{\cal U}'_2,\dots$ be the ultrametric spaces corresponding to the subtrees rooted at 
children of the root of ${\cal U}'$, and let $n_i:=|U'_i|$, where ${\cal U}'_i = (U'_i,d_{|U'_i})$. We sort 
the subtrees so that $n_1\ge n_2\ge\dots$. There are two cases, defined by Lemma~\ref{lem:sqrt}: 
either $\sqrt{n_1}+\sqrt{n_2}\ge \sqrt{n}$, or there exists $\ell\ge3$ such that $\ell\cdot\sqrt{n_\ell}\ge\sqrt{n}$.  
In the former case, we let $\ell:=2$ and we call this the \emph{binary case}. The latter case is called 
the \emph{balanced case} because the proof will invoke the same lower bound $\alpha\log n_\ell$ 
for each of the first $\ell$ subtrees.

In the balanced case, if additionally $\log\ell\ge 2\alpha n_\ell$, then we set $U_i = \{x_i\}$ for all 
$i\in\{1,2,\dots,\ell\}$, where $x_i\in U'_i$ is an arbitrary point. We call this the special case of the 
balanced case the \emph{uniform case}. In all other cases, we set $U_i\subseteq U'_i$ to be the 
set of points of the subspace ${\cal U}_i$ of ${\cal U}'_i$ that the induction hypothesis stipulates,
for all $i\in\{1,2,\dots,\ell\}$. Finally, we let $U:=\bigcup_{i=1}^\ell U_i$. Notice that in the uniform 
case, ${\cal U}$ is an $\ell$-point uniform metric space.

For convenience, we will use in this proof the convention that each request specifies the set of 
points where the algorithm must \emph{not} be, i.e., the complement of the sets we used earlier. 
This means that if we issue an inductive request sequence in some $U_i$, then any algorithm 
located in some $U_j$ for $j\ne i$ does not incur any cost. Note that the entire proof is written
in a manner that restricts the request sequence to the ${\cal U}$ subspace. To apply the lower
bound to the original ${\cal U}'$ space, we need under the above convention to add to each
request in the sequence all the points in $U'\setminus U$.

If $|U_i|\ge 2$ (so that $\diam({\cal U}_i)>0$), let $\cD_i$ be the distribution on request sequences 
in $U_i$ induced by the induction hypothesis. In this case, we call a \emph{chunk in $U_i$} a 
concatenation of $\frac{\diam({\cal U})}{2\diam({\cal U}_i)}$ independent random request sequences 
from $\cD_i$. We call each such sequence from $\cD_i$ a \emph{subchunk}. If $|U_i|=1$, we define 
a \emph{chunk in $U_i$} to be either the empty request sequence or a single request to the singleton 
point in $U_i$, chosen with probabilitiy $1/2$ each. A random request sequence in $\cD$ is obtained 
by repeating the following $2\ell$ times independently: Choose $i\in\{1,\dots,\ell\}$ uniformly at random, 
then issue a chunk in $U_i$. See Figure~\ref{fig:chunks} for an illustration of these definitions.

\begin{figure}
	\begin{center}
		\includegraphics[width=0.35\textwidth]{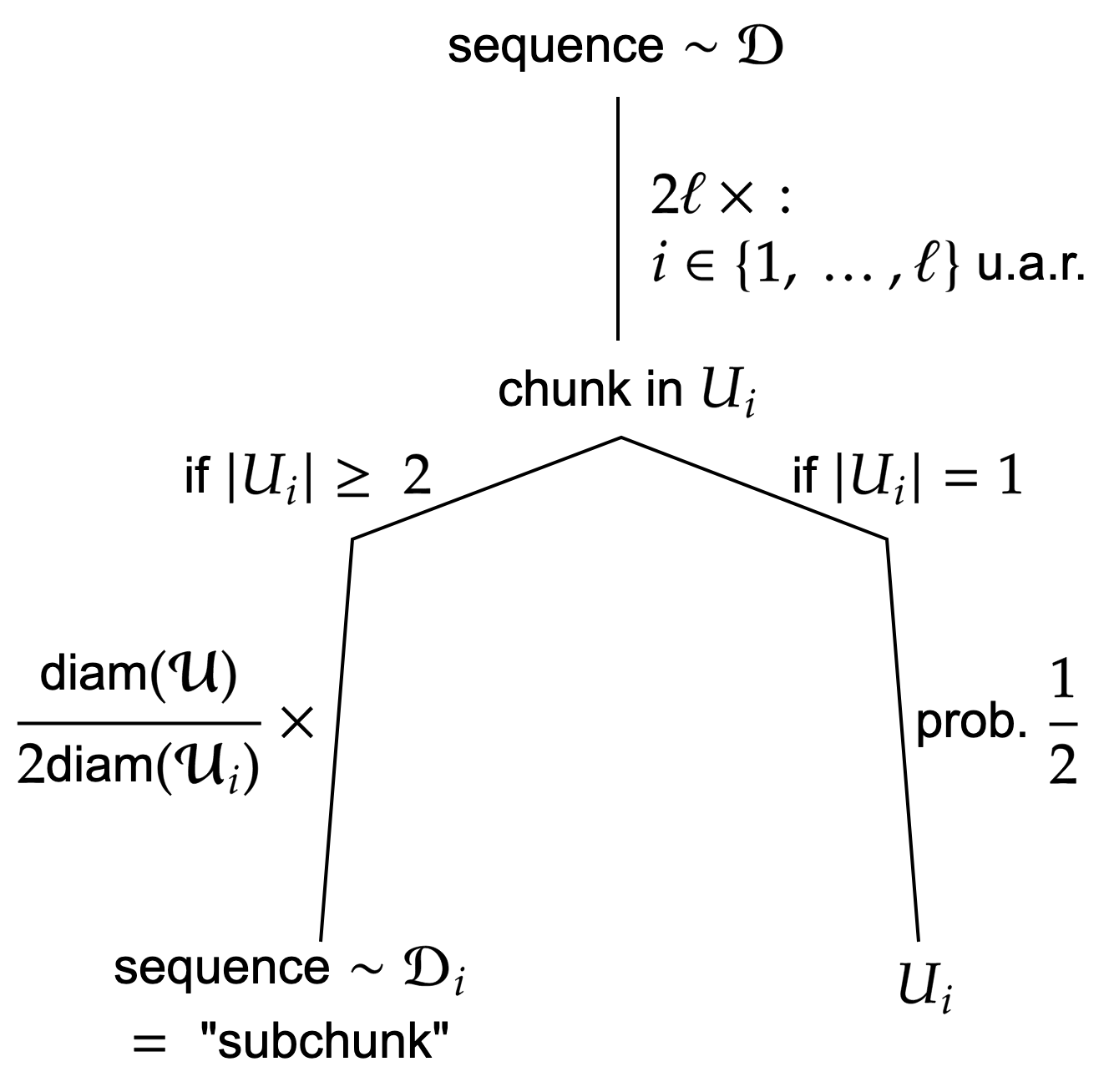}
	\end{center}
	\caption{Construction of sequences in $\cD$ from chunks and subchunks in the universal lower bound.}\label{fig:chunks}
\end{figure}

\subsection{\texorpdfstring{Proof of inequalities~\eqref{eq:universalCost}}{Proof of inequalities}}\label{sec: universal cost}

For the first inequality, it suffices to show that we can charge expected cost at least $\frac{\diam({\cal U})}{2\ell}$ 
to each chunk. Note that if the algorithm invokes the escape price of $2\diam({\cal U})$, the true cost on all 
subsequent chunks would be $0$. But in this case we charge only cost $\diam({\cal U})$ instead of the full
escape price $2\diam({\cal U})$ to the chunk on which the escape price is invoked. We can then use the
uncharged part of the escape price to charge $\frac{\diam({\cal U})}{2\ell}$ to each subsequent chunk. This charging
scheme does not overcharge the escape price, and guarantees that we can charge the claimed amount even 
to chunks starting after the algorithm has already escaped. It remains to consider chunks that are issued when 
the algorithm has not yet escaped.

Let ${\cal U}_i$ be the subtree where the online algorithm $\alg$ is located before a chunk is issued. With 
probability $\frac{1}{\ell}$, the chunk is issued in $U_i$, so it suffices to show that conditioned on this being 
the case, $\alg$ has to pay at least $\diam({\cal U})/2$ in expectation on this chunk. If $|U_i|=1$, this is trivial
(as $\alg$ has to move elsewhere in $U$). If $|U_i|\ge 2$, we argue that we can charge an expected cost 
of $\diam({\cal U}_i)$ to each of the $\frac{\diam({\cal U})}{2\diam({\cal U}_i)}$ subchunks of the chunk. If 
$\alg$ stays inside $U_i$ during the chunk, we charge to each subchunk the cost that is actually suffered 
on it. 

Now suppose that $\alg$ invokes the escape price $2\diam({\cal U})$ or switches to some $U_j$ with $j\ne i$.
We call the or-combination of these events the \emph{escape-or-switch} event. Such a move incurs a one-time 
cost of $\diam({\cal U})$. The reason is that this is the cost of switching from $U_i$ to $U_j$. Recall that in 
the case of escape we charge only $\diam({\cal U})$ of the actual escape cost $2\diam({\cal U})$ to this 
chunk, and the true cost on the remaining subchunks would be $0$. Of this cost $\diam({\cal U})$, we charge 
$2\diam({\cal U}_i)$ as escape cost to the subchunk where the escape-or-switch happens, and we charge 
$\diam({\cal U}_i)$ to all remaining at most $\frac{\diam({\cal U})}{2\diam({\cal U}_i)}-1$ subchunks belonging 
to this chunk (whose true cost would have been $0$). So, in total we charge at most 
$\frac{\diam({\cal U})}{2}+\diam({\cal U}_i)\le \diam({\cal U})$ for an escape-or-switch, and in particular we do 
not overcharge. 

Given this charging scheme, the induction hypothesis implies that the expected cost charged to each subchunk 
is at least $\diam({\cal U}_i)$. As there are $\frac{\diam({\cal U})}{2\diam({\cal U}_i)}$ subchunks in the chunk, 
the expected cost of a chunk in the subtree ${\cal U}_i$ where $\alg$ is located at the start of the chunk is at 
least $\diam({\cal U})/2$, as desired. This completes the proof of the first inequality in~\eqref{eq:universalCost}.

To see the second inequality in~\eqref{eq:universalCost}, consider the following algorithm (which is actually an 
online algorithm and therefore shows that our analysis for the first inequality was tight): Let ${\cal U}_i$ be the 
subtree where the algorithm is currently located. If $|U_i|\ge 2$, then the algorithm stays in $U_i$ and plays 
according to the induction hypothesis. If $|U_i|=1$, then the algorithm switches to a different $U_j$ only if $U_i$ 
is requested. For this algorithm, a chunk in the subtree ${\cal U}_i$ where the algorithm is located incurs expected 
cost $\diam({\cal U})/2$. As it happens twice in expectation during a random $\rho\sim\cD$ that a chunk is issued 
in the subtree where the algorithm is located, the second inequality in~\eqref{eq:universalCost} follows.

It remains to prove inequality~\eqref{eq:universalCR}. We will do so separately for the uniform, balanced 
non-uniform, and binary case.

\subsection{The uniform case}\label{sec: universal uniform}

We first consider the case that there exists $\ell\ge 3$ such that $\ell\cdot\sqrt{n_\ell}\ge\sqrt{n}$ and additionally 
$\log\ell\ge 2\alpha \log n_\ell$. Recall that in this case we choose ${\cal U}$ to be an $\ell$-point uniform metric. 
The request sequence $\rho_1\dots\rho_h$ for $(\rho_1,\dots,\rho_h)\sim\cD^h$ has that each $\rho_j$ is a sequence of at most $2\ell$ singleton point sets from $U$ that are chosen independently and uniformly 
at random.
For large $h$, this is precisely the sequence 
from~\cite{BLS87} that yields a lower bound of $H_{\ell-1}>\log(\ell-1)$, where $H_j=1+\frac 1 2+\frac 1 3+\cdots+\frac 1 j$, 
on the competitive ratio of $\MTS$ 
(and, in fact, $\MSS$) in uniform metric spaces.\footnote{The proof follows from a coupon collector argument.} The expected online cost on this 
sequence is precisely $h\cdot\diam({\cal U})$. 
We thus have for large $h$ that
\begin{align*}
        \frac{h\cdot\diam({\cal U})}{\E_{(\rho_1,\dots,\rho_h)\sim\cD^h}\left[c_{\opt,s}(\rho_1\rho_2\dots\rho_h)\right]} 
&\ge \log (\ell-1)\\
&\ge 2\alpha\log\ell + \alpha\log n_\ell \\
&= 2\alpha\log(\ell\sqrt n_\ell)\\
&\ge \alpha\log n,
\end{align*}
where the second inequality holds for a sufficiently small constant $\alpha$ and uses the assumption 
$\log\ell\ge 2\alpha \log n_\ell$, and the last inequality uses $\ell\sqrt{n_\ell}\ge \sqrt{n}$. This proves
Inequality~\eqref{eq:universalCR} in the uniform case.

\subsection{The balanced non-uniform case}\label{sec: universal balanced}

Here, we consider the case that there exists $\ell\ge 3$ such that $\ell\cdot\sqrt{n_\ell}\ge\sqrt{n}$ and 
$\log\ell< 2\alpha \log n_\ell$. In particular, we have $n_\ell>1$, so none of the $U_i$-s is a singleton.

Let
\begin{align}
\mu:=\left\lceil\frac{\alpha\cdot\log^2 n_\ell}{\log\ell}\right\rceil\in\left[\log\ell, \frac{2\alpha\cdot\log^2 n_\ell}{\log\ell}\right]\label{eq:universalBalancedm}
\end{align}
where the lower bound uses $\log\ell< 2\alpha \log n_\ell$ and holds for sufficiently small $\alpha$ 
(namely, $\alpha\le 1/4$), and the upper bound follows from the lower bound $\log\ell\ge \log3\ge 1$.

Note that the request sequence $\rho_1\dots\rho_h$ is simply a concatenation of $2h\ell$ independently 
sampled chunks, each in a random $U_i$. Let $k,j\in\N_0$ be the numbers determined by the relations 
$k\mu+j=h$ and $j<\mu$.  We decompose $\rho_1\dots\rho_h$ into $k+1$ \emph{phases} as follows: 
For $m=1,\dots, k$, the $m$th phase is the subsequence $\rho_{(m-1)\mu+1}\rho_{(m-1)\mu+2}\dots\rho_{m\mu}$. 
The last phase is the remaining suffix $\rho_{h-j+1}\rho_{h-j+2}\dots\rho_{h}$. Thus, the first $k$ phases 
consist of $2\mu\ell$ chunks each, and the last phase consists of $2j\ell$ chunks.

Consider the following type of offline algorithm, whose cost we will use as an upper bound on the optimal 
cost: At the beginning of a phase, move to the subspace ${\cal U}_i$ where the smallest number of 
chunks is played during the phase, and stay in ${\cal U}_i$ while serving all requests of that phase. We 
will separately analyze the offline \emph{switching cost} for switching between different subspaces ${\cal U}_i$, 
and the offline \emph{local cost} incurred within the ${\cal U}_i$-s.

At the start of each of the $k+1$ phases, the offline algorithm switches to a different ${\cal U}_i$ for cost 
$\diam({\cal U})$ with probability at most $1-\frac{1}{\ell}$. Thus, the expected offline switching cost is at 
most
\begin{align}
       (k+1)\left(1-\frac{1}{\ell}\right)\diam({\cal U})
&\le \left(\frac{h}{\mu}+1\right)\left(1-\frac{1}{\ell}\right)\diam({\cal U})\notag\\
&\le \frac{h}{\mu}\diam({\cal U})\notag\\
&\le \frac{\log\ell}{\alpha\log^2 n_\ell} \cdot h\cdot\diam({\cal U}),\label{eq:universalSwitch}
\end{align}
where the second inequality holds for any $h\ge \mu\ell$, and the last inequality follows by definition of $\mu$.

Let $h_i$ be the total number of chunks played in $U_i$ while the offline algorithm resides in $U_i$. Our 
first goal is to obtain an upper bound on $\E\left[\sum_i h_i\right]$, i.e., the expected number of chunks 
contributing to the offline local cost. In each of the first $k$ phases, the total number of chunks is $2\mu\ell$, 
and thus by Lemma~\ref{lem:ballsBinsMin} the expected number of chunks played in the subspace in which 
the offline algorithm resides during that phase is at most
\begin{align*}
2\mu - c\sqrt{\mu\log \ell}
\end{align*}
for some constant $c>0$. In the last phase, the according quantity is at most $2j$. Summing over all phases, 
we get
\begin{align}
\E\left[\sum_i h_i\right] &\le k\cdot\left(2\mu-c\sqrt{\mu\log \ell}\right) + 2j\notag\\
&\le 2k\mu+2j - \frac{c}{2}\sqrt\frac{\log\ell}{\mu}\left(k\mu+j\right)\notag\\
&= 2h\cdot\left(1 - \frac{c}{4}\sqrt\frac{\log\ell}{\mu}\right)\notag\\
&\le 2h\cdot\left(1 - \frac{c}{\sqrt{32\alpha}}\cdot \frac{\log\ell}{\log n_\ell}\right)\label{eq:universalEhi}
\end{align}
where the second inequality holds for $h\ge \mu$, since then $k\mu\ge\mu\ge j$, the equation uses $k\mu+j=h$, 
and the inequality uses the upper bound in Inequality~\eqref{eq:universalBalancedm}.

Since each chunk in $U_i$ consists of $\frac{\diam({\cal U})}{2\diam({\cal U}_i)}$ independent samples 
from $\cD_i$, the subsequence of $\rho_1\dots\rho_h$ that contributes to the local offline cost in ${\cal U}_i$ 
consists of $\frac{h_i\diam({\cal U})}{2\diam({\cal U}_i)}$ independent samples from $\cD_i$. Although these 
usually occur as several strings with gaps between them, we can still apply the induction hypothesis on 
${\cal U}_i$ to bound the offline local cost in ${\cal U}_i$, because the offline algorithm can always return 
to the point where it left off when switching back to $U_i$. The induction hypothesis shows that the expected 
offline local cost within ${\cal U}_i$ is at most
\begin{align*}
&\frac{\diam({\cal U})}{2\diam({\cal U}_i)}\cdot\bigg(\Pr[h_i\ge h({\cal U}_i)]\cdot \E[h_i\mid h_i\ge 
h({\cal U}_i)]\cdot \frac{\diam({\cal U}_i)}{\alpha\log n_i}\\
& \qquad\qquad\qquad+ \Pr[h_i<h({\cal U}_i)]\cdot \E[h_i\mid h_i<h({\cal U}_i)]\cdot \diam({\cal U}_i)\bigg)\\
&\le \frac{\diam({\cal U})}{2}\left(\frac{\E[h_i]}{\alpha\log n_\ell} \quad + \quad h({\cal U}_i)\right)\\
&\le \frac{\E[h_i]\diam({\cal U})}{2\alpha\log n_\ell}\cdot(1+o(1)),
\end{align*}
where the first inequality uses that $n_i\ge n_\ell$ (since $i\le \ell$), and $o(1)$ describes a term that tends 
to $0$ as $h\to\infty$. (This is valid since $\E[h_i]\to\infty$ as $h\to\infty$, but $h({\cal U}_i)$ is fixed.) Combined 
with Equation~\eqref{eq:universalEhi} and the upper bound of Inequality~\eqref{eq:universalSwitch} on the 
switching cost, we get for $h$ sufficiently large that
\begin{align*}
\E\left[c_{\opt,s}(\rho_1\rho_2\dots\rho_h)\right] &\le \frac{\log\ell}{\alpha\log^2 n_\ell} \cdot h\cdot \diam({\cal U}) \quad+\quad \frac{h\cdot\diam({\cal U})}{\alpha\log n_\ell}\left(1 - \frac{c}{6\sqrt{\alpha}} \frac{\log\ell}{\log n_\ell}\right)\\
&\le \frac{h\cdot\diam({\cal U})}{\alpha\log n_\ell}\left(\frac{\log\ell}{\log n_\ell} + 1-\frac{c}{6\sqrt{\alpha}}\frac{\log \ell}{\log n_\ell}\right)\\
&\le \frac{h\cdot\diam({\cal U})}{\alpha\log n_\ell}\left(1-\frac{2\log \ell}{\log n_\ell}\right),
\end{align*}
where the last inequality holds provided the constant $\alpha$ is sufficiently small.
From this, we conclude~\eqref{eq:universalCR} via
\begin{align*}
\frac{h\cdot\diam({\cal U})}{\alpha\cdot \E\left[c_{\opt,s}(\rho_1\rho_2\dots\rho_h)\right]}&\ge \frac{\log n_\ell}{1-\frac{2\log \ell}{\log n_\ell}}\\
&\ge \log n_\ell+2\log\ell\\
&= 2\log (\ell \sqrt{n_\ell})\\
&\ge \log n,
\end{align*}
where the second inequality uses $\frac{1}{1-x}\ge 1+x$ for all $x<1$, and the last inequality 
uses $\ell\sqrt{n_\ell}\ge \sqrt n$.

\subsection{The binary case}\label{sec: universal binary}

Finally, we consider the case that $\sqrt{n_1}+\sqrt{n_2}\ge\sqrt{n}$. Recall that $n_1\ge n_2$. 
Since $\alpha\log n\le 2\alpha\log(\sqrt{n_1}+\sqrt{n_2})\le 2\alpha\log(2\sqrt{n_1})=\alpha\log(4n_1)$, 
we may assume that $\alpha\log(4n_1)>1$, since otherwise Inequality~\eqref{eq:universalCR} follows 
trivially from the second inequality of~\eqref{eq:universalCost}. By choosing $\alpha\le 1/\log(16)$, this 
also implies $n_1\ge 4$, which together implies
\begin{align}
\log n_1>\frac{1}{2\alpha}.\label{eq:univeralBinaryAssumption}
\end{align}
Let
\begin{align}
\delta_1&:=\frac{\max\left\{\frac{1}{\sqrt{\alpha}}, \log\frac{n_1}{n_2}\right\}}{\log n_1}\in(0,1]\notag\\
\delta_2&:= 1-(1-\delta_1)\frac{\log n_2}{\log n_1}=\frac{\log\frac{n_1}{n_2}}{\log n_1} + \delta_1\frac{\log n_2}{\log n_1}\in[\delta_1,2\delta_1]\notag\\
\mu &:= \left\lceil\frac{8\alpha\log n_1}{\delta_1}\right\rceil\in\left[4, \frac{10\alpha\log n_1}{\delta_1}\right],\label{eq:universalBinaryMu}
\end{align}
where the bounds on $\delta_1$ and $\mu$ hold due to Inequality~\eqref{eq:univeralBinaryAssumption} for 
sufficiently small $\alpha$, and the bounds on $\delta_2$ hold because $n_1\ge n_2$.

Let $k,j\in\N_0$ be the numbers satisfying $k\mu+j=h$ and $j< \mu$. We decompose the sequence 
$\rho_1\dots\rho_h$ into $k+1$ phases in the same way as in the balanced non-uniform case. So each 
of the first $k$ phases consists of $4\mu$ chunks, of which $2\mu$ are played in each of $U_1$ and $U_2$ 
in expectation. We consider the following offline algorithm: By default, it stays in $U_1$, but for any of the $k$ 
complete phases where the number of chunks in $U_2$ is at most 
$(1-\delta_1)\frac{\log n_2}{\log n_1}\cdot2\mu=(1-\delta_2)2\mu$, the algorithm moves to $U_2$ at the beginning 
of the phase and back to $U_1$ at the end of the phase. Note that this is feasible even if $|U_2|=1$, since then 
$\log n_2=0$ and the algorithm moves to $U_2$ only in phases that have \emph{no} chunk in $U_2$. In phase 
$k+1$, the offline algorithm is in $U_1$.

Denote by $p>0$ the probability of switching to $U_2$ in a given phase. (Clearly, this probability is the same for 
each of the first $k$ phases.) In expectation, in each of the first $k$ phases the offline algorithm pays 
$2p\diam({\cal U})$ for switching to $U_2$ and back, and it may pay an additional $\diam({\cal U})$ to switch to 
$U_1$ before the first phase. Thus, the expected offline switching cost is at most
\begin{align}
(2pk+1)\diam({\cal U})&\le \frac{2ph\diam({\cal U})}{\mu}(1+o(1))\notag\\
&\le \frac{\delta_1ph\diam({\cal U})}{4\alpha\log n_1}(1+o(1)),\label{eq:universalBinarySwitching}
\end{align}
where $o(1)$ denotes again a term that tends to $0$ as $h\to\infty$, and the second inequality uses the 
definition of $\mu$.

To bound the offline local cost, let $h_1$ and $h_2$ be again the number of chunks played in $U_1$ 
and $U_2$, respectively, while the offline algorithm is in that subspace. In each of the first $k$ phases, 
$h_1$ grows by at most $(1-p)2\mu$ in expectation and $h_2$ grows by at most 
$p(1-\delta_1)\frac{\log n_2}{\log n_1}\cdot2\mu$ in expectation. In phase $k+1$, $h_1$ grows by another 
$2j<2\mu$ in expectation and $h_2$ stays put. Thus, for the entire request sequence $\rho_1\dots\rho_h$, 
we have
\begin{align*}
\E[h_1]&\le k(1-p)2\mu + 2\mu\le (1-p) 2h(1+o(1))\\
\E[h_2]&\le kp(1-\delta_1)\frac{\log n_2}{\log n_1}2\mu \le 2hp(1-\delta_1)\frac{\log n_2}{\log n_1}.
\end{align*}
The induction hypothesis shows that the expected offline local cost within a non-singleton $U_i$ is at most
\begin{align*}
&\frac{\diam({\cal U})}{2\diam({\cal U}_i)}\cdot\bigg(\Pr[h_i\ge h({\cal U}_i)]\cdot \E[h_i\mid h_i\ge 
h({\cal U}_i)]\cdot \frac{\diam({\cal U}_i)}{\alpha\log n_i}\\
& \qquad\qquad\qquad+ \Pr[h_i<h({\cal U}_i)]\cdot \E[h_i\mid h_i<h({\cal U}_i)]\cdot \diam({\cal U}_i)\bigg)\\
&\le \frac{\diam({\cal U})}{2}\left(\frac{\E[h_i]}{\alpha\log n_i} \quad + \quad h({\cal U}_i)\right)\\
&\le \frac{\E[h_i]\diam({\cal U})}{2\alpha\log n_i}\cdot(1+o(1)),
\end{align*}
If $U_i$ is a singleton, which is only possible for $i=2$, then the local cost within $U_2$ is $0$.

Combined with the bounds on $\E[h_1]$ and $\E[h_2]$ and the bound~\eqref{eq:universalBinarySwitching} 
on the switching cost, we can bound the total expected offline cost by
\begin{align}
&\E\left[c_{\opt,s}(\rho_1\rho_2\dots\rho_h)\right] \notag\\
&\le \left(\frac{\delta_1ph\diam({\cal U})}{4\alpha\log n_1} + 
     \frac{(1-p) h\diam({\cal U})}{\alpha\log n_1} + \frac{hp(1-\delta_1)\diam({\cal U})}{\alpha\log n_1}\right)(1+o(1))\notag\\
&\le \frac{h\diam({\cal U})}{\alpha\log n_1}\left(\frac{\delta_1p}{4} + 1-p\delta_1\right)(1+o(1))\notag\\
&\le \frac{h\diam({\cal U})}{\alpha\log n_1}\left(1-\frac{p\delta_1}{2}\right),\notag
\end{align}
where the last inequality holds for $h$ sufficiently large so that the $o(1)$ term is very small. Therefore,
\begin{align}
        \frac{h\diam({\cal U})}{\alpha\cdot\E\left[c_{\opt,s}(\rho_1\rho_2\dots\rho_h)\right]} 
&\ge \frac{\log n_1}{1-\frac{p\delta_1}{2}}\notag \\
&\ge \log n_1 + \frac{p}{2\sqrt{\alpha}},\label{eq:universalBinaryCR}
\end{align}
where the last inequality uses $\delta_1\ge \frac{1}{\sqrt{\alpha}\log n_1}$.

To conclude the lemma, we need to lower bound the latter quantity by $\log n$, which requires a lower bound on $p$. By Lemma~\ref{lem:binom}, there is a constant $\lambda>0$ such that
\begin{align*}
p\ge \lambda e^{-\delta_2^22\mu/\lambda}.
\end{align*}
Using $\delta_2\le 2\delta_1$ and the upper bound on $\mu$ from~\eqref{eq:universalBinaryMu}, we get
\begin{align*}
p&\ge \lambda \exp\left(-\frac{80\delta_1\alpha\log n_1}{\lambda}\right)\\
&\ge \lambda\exp\left(-\frac{80\sqrt{\alpha}}{\lambda} - \frac{80\alpha}{\lambda}\log\frac{n_1}{n_2}\right)\\
&\ge 4\sqrt{\alpha}\sqrt{\frac{n_2}{n_1}},
\end{align*}
where the second inequality uses the definition of $\delta_1$ and $\max\{x,y\}\le x+y$ for $x,y\ge 0$, and 
the last inequality holds provided the constant $\alpha$ is sufficiently small (i.e., satisfying $\frac{80\sqrt{\alpha}}{\lambda}\le \log\frac{\lambda}{4\sqrt{\alpha}}$ and $\frac{80\alpha}{\lambda}\le \frac{1}{2}$). Plugging this bound on $p$ into~\eqref{eq:universalBinaryCR}, we get
\begin{align*}
\frac{h\diam({\cal U})}{\alpha\cdot\E\left[c_{\opt,s}(\rho_1\rho_2\dots\rho_h)\right]} &\ge \log n_1 + 2\sqrt\frac{n_2}{n_1}\\
&\ge \log n_1 + 2\log\left(1+\sqrt\frac{n_2}{n_1}\right)\\
&=2\log(\sqrt{n_1}+\sqrt{n_2})\\
&\ge \log n,
\end{align*}
completing the proof of Lemma~\ref{lem:universalMain}.

\section*{Acknowledgements}
We wish to thank Yossi Azar, Yair Bartal, and Manor Mendel for pointing out relevant 
references, and the anonymous reviewers of STOC '23 for useful suggestions.

\appendix
\section{Appendix}

\subsection{Metric spaces}\label{sec: metrics}

A {\em metric space} ${\cal M}$ is a pair $(M,d)$, where $M$ is a set and $d:M\times M\rightarrow [0,\infty)$
such that ($i$) $d$ is symmetric: $d(x,y)=d(y,x)$ for all $x,y\in M$, ($ii$) $d(x,y) = 0$ iff $x=y$, and
($iii$) $d$ satisfies the triangle inequality: $d(x,z)\le d(x,y)+d(y,z)$ for all $x,y,z\in M$. Here we are
concerned primarily with finite metric spaces ($M$ is a finite set). 

Let ${\cal M} = (M,d)$ be a metric space. The {\em diameter} of ${\cal M}$,
denoted $\diam({\cal M})$ is the supremum over $x,y\in M$ of $d(x,y)$.
A metric space ${\cal M}' = (M',d')$ is a subspace of ${\cal M} = (M,d)$ iff $M'\subset M$
and $d'$ is the restriction $d_{|M'}$ of $d$ to $M'$.

Let ${\cal M} = (M,d)$ and ${\cal M}' = (M',d')$ be two finite metric spaces with $|M| = |M'|$, and let 
$\phi: M\rightarrow M'$ be a bijection. The {\em Lipschitz constant} of $\phi$ is 
$$
\|\phi\|_{\Lip} = \max_{x\ne y\in M} \frac{d'(\phi(x),\phi(y))}{d(x,y)}.
$$
This is the maximum relative expansion of a distance under $\phi$. Similarly, we can define the Lipschitz 
constant of the inverse mapping $\phi^{-1}$ as the maximum over the reciprocal expressions. It measures 
the maximum relative contraction of a distance under $\phi$. The (bi-Lipschitz) {\em distortion} of 
$\phi$ is $\|\phi\|_{\Lip} \cdot\|\phi^{-1}\|_{\Lip}$. As the term suggests, it measures the relative maximum
distortion of distance, up to uniform scaling, of the metric space ${\cal M}$ under the mapping $\phi$.
The {\em bi-Lipschitz distortion} between ${\cal M}$ and ${\cal M'}$ is the minimum over bijections 
$\phi:M\rightarrow M'$ of $\|\phi\|_{\Lip} \cdot\|\phi^{-1}\|_{\Lip}$. Notice that this number is always
at least $1$. If it is equal to $1$, we say that the two spaces are isometric and that the minimizing
$\phi$ is an isometry.

A metric space ${\cal U} = (U,d)$ is called an {\em Urysohn universal} space iff it is separable
(contains a countable dense subset) and complete (every Cauchy sequence in $U$ has a limit
in $U$) and satisfies the following property. For every finite metric space ${\cal M} = (M,d')$
and for every $x\in M$, for every injection $\phi: M\setminus\{x\}\rightarrow U$ which is
an isometry (between the relevant subspaces of ${\cal M}$ and ${\cal U}$), there exists an extension 
$\phi'$ of $\phi$ to $x$ which is an isometry (between ${\cal M}$ and the relevant subspace of 
${\cal U}$). Urysohn~\cite{Ury27} proved that an Urysohn universal space exists and is unique 
up to isometry.

A metric space ${\cal U} = (U,d)$ is called an {\em ultrametric space} iff $d$ satisfies an
inequality stronger than the triangle inequality, namely, the following. For every 
three points $x,y,z\in U$, it holds that $d(x,z)\le\max\{d(x,y),d(y,z)\}$. A metric
$d$ that satisfies this inequality is called an ultrametric. Every finite ultrametric
space can be represented by a node-weighted rooted tree structure which is 
called a {\em hierarchically separated tree}, abbreviated HST. Conversely, every
HST represents a finite ultrametric space. The weights on the internal nodes of
an HST are positive and non-increasing along any path from root to leaf. The leaves 
of the HST all have weight $0$, and they represent the points of the space. The 
distance between two points is the weight of their least common ancestor. An HST 
in which the weight of each internal node is at least a factor of $q$ larger than the 
weight of any of its children is called a $q$-HST. Every HST is in particular a $1$-HST.
Thus, the notions of a finite ultrametric space and of a $1$-HST are equivalent. For
every finite ultrametric space there exists an isometric $1$-HST, and vice versa.
(Note: an HST is a representation of a metric over the set of its leaves; the internal
nodes are not points in this metric space.)

\subsection{Some inequalities}

\begin{theorem}[Berry-Esseen Inequality]\label{thm: BE}
There exists an absolute constant $c$ such that the following holds.
Let $X_1,X_2,\dots,X_n$ be i.i.d. random variables, each having expectation $0$,
standard deviation $\sigma > 0$, and third moment $\rho < \infty$. Let $F$ denote
the cumulative distribution function of $\frac{1}{\sigma\sqrt{n}} \sum_{i=1}^n X_i$,
and let $\Phi$ denote the cumulative distribution function of the standard normal 
distribution ${\cal N}(0,1)$ with expectation $0$ and standard deviation $1$. Then, 
for all $x\in\R$, $\left|F(x) - \Phi(x)\right|\le \frac{c\rho}{\sigma^3\sqrt{n}}$.
\end{theorem}

The following two Lemmas were proved in~\cite{BBM01}.
\begin{lemma}[{\cite[Proposition 11]{BBM01}}]\label{lem:sqrt}
Let $(n_i)_{i\ge 1}$ be a non-increasing sequence of positive real numbers such that 
$n:=\sum_{i}n_i<\infty$. Then either $\sqrt{n_1}+\sqrt{n_2}\ge \sqrt{n}$ or there exists 
$\ell\ge3$ such that $\ell\cdot\sqrt{n_\ell}\ge\sqrt{n}$.
\end{lemma}

\begin{lemma}[{\cite[Lemma 30]{BBM01}}]\label{lem:binom}
There exists a constant $\lambda\in(0,1]$ such that for any binomial random variable $X$ 
with $p\le 0.5$ and mean $\mu\ge 4$ and any $\delta\in[0,1]$,
\begin{align*}
	\Pr[X\le(1-\delta)\mu]\ge \lambda e^{-\delta^2\mu/\lambda}.
\end{align*}
\end{lemma}

\begin{lemma}\label{lem:ballsBinsMin}
Consider the experiment of placing $m$ balls independently and uniformly at random into $n\ge 2$ 
bins, where $m\ge n\ln n$. Let $X_i$ be the number of balls in the $i$th bin. Then
\begin{align*}
	\E\left[\min_i X_i\right]\le \frac{m}{n}-c\sqrt\frac{m\ln n}{n}
\end{align*}
for some global constant $c>0$ (which does not depend on $m$ or $n$).
\end{lemma}

\begin{proof}
Let $c\in (0,1)$ be a constant to be determined later. Let $\mu = \E[X_1] = \frac m n$. Note that
$\min_i X_i\le \mu$ always. We consider two cases. Firstly, suppose that $\mu < 4$. Thus, 
$\ln n\le \frac m n < 4$. In particular, $n < e^4$ and $m < 4e^4$. As 
$\Pr[X_1 = 0] = \left(1-\frac 1 n\right)^m \ge \eta$ for a constant $\eta > 0$, we have that
$\E\left[\min_i X_i\right]\le (1-\eta)\mu\le \frac{m}{n}-c\sqrt\frac{m\ln n}{n}$ for a sufficiently
small constant $c > 0$.

Otherwise, let $\delta = c\sqrt{\frac{n\ln n}{m}}\in (0,c)$. We shall prove that with the right choice 
of constant $c$, $\E\left[\min_i X_i\right]\le (1-\delta)\mu$. Let $E_i$ denote the event that 
$X_i > (1-a\delta)\mu$, for some constant $a>0$ to be determined later. Notice that 
$\Pr[E_i\mid E_1\wedge E_2\wedge \cdots \wedge E_{i-1}]\le\Pr[E_i]$. Therefore,
$\Pr[\min_i X_i > (1-a\delta)\mu] \le \left(\Pr[E_1]\right)^n$. Also let $\lambda\in (0,1]$ be the constant
stipulated in Lemma~\ref{lem:binom}. By that lemma (which we can apply provided $ac\le 1$, so 
that $a\delta\le 1$),
$$
\Pr[\bar{E_1}] \ge \lambda e^{-\frac{(a\delta)^2}{\lambda}\mu} = \lambda n^{-\frac{(ac)^2}{\lambda}}.
$$
Therefore, 
$$
\Pr[\min_i X_i > (1-a\delta)\mu] < 
   \left(1 - \lambda n^{-\frac{(ac)^2}{\lambda}}\right)^n\le e^{-\lambda n^{1-\frac{(ac)^2}{\lambda}}}.
$$
As always $\min_i X_i\le \frac m n$,  we have that $\E\left[\min_i X_i\right]\le (1 - (1-\epsilon)a\delta)\frac m n$,
where $\epsilon = e^{-\lambda n^{1-\frac{(ac)^2}{\lambda}}}$. We need to set $a$ so that $(1-\epsilon)a\ge 1$ 
and $c$ so that $c \le \frac 1 a$, and then the lemma follows in this case. For example, we can set 
$a=\frac{1}{1-e^{-\lambda}}$ and $c \le \frac{\sqrt{\lambda}}{a}$. Thus, $\frac{(ac)^2}{\lambda} \le 1$,
and $\epsilon \le e^{-\lambda}$.
\end{proof}

\bibliographystyle{plain}
\bibliography{biblio}

\begin{thebibliography}{10}

\bibitem{ACNN11}
A.~Andoni, M.~S. Charikar, O.~Neiman, and H.~L. Nguyen.
\newblock Near linear lower bound for dimension reduction in $\ell_1$.
\newblock In {\em Proc.of the 52nd Ann. IEEE Symp. on Foundations of Computer
  Science}, pages 315--323, 2011.

\bibitem{ABF93}
B.~Awerbuch, Y.~Bartal, and A.~Fiat.
\newblock {H}eat {\&} {D}ump: Competitive distributed paging.
\newblock In {\em Proc. of the 34th Ann. IEEE Symp. on Foundations of Computer
  Science}, pages 22--31, 1993.

\bibitem{BBMN11}
N.~Bansal, N.~Buchbinder, A.~Madry, and J.~Naor.
\newblock A polylogarithmic-competitive algorithm for the $k$-server problem.
\newblock In {\em Proc. of the 52nd Ann. IEEE Symp. on Foundations of Computer
  Science}, pages 267--276, October 2011.

\bibitem{BBN07}
N.~Bansal, N.~Buchbinder, and J.~Naor.
\newblock A primal-dual randomized algorithm for weighted paging.
\newblock In {\em Proc. of the 48th Ann. IEEE Symp. on Foundations of Computer
  Science}, pages 507--517, 2007.

\bibitem{BBN10a}
N.~Bansal, N.~Buchbinder, and J.~Naor.
\newblock Towards the randomized $k$-server conjecture: A primal-dual approach.
\newblock In {\em Proc. of the 21st Ann. ACM-SIAM Symp. on Discrete
  Algorithms}, pages 40--55, 2010.

\bibitem{BC21b}
N.~Bansal and C.~Coester.
\newblock Online metric allocation and time-varying regularization.
\newblock In {\em Proc. of the 30th Ann. European Symp. on Algorithms}, pages
  13:1--13:13, 2022.

\bibitem{BBBT97}
Y.~Bartal, A.~Blum, C.~Burch, and A.~Tomkins.
\newblock Polylog($n$)-competitive algorithm for metrical task systems.
\newblock In {\em Proc. of the 29th Ann. ACM Symp. on Theory of Computing},
  pages 711--719, 1997.

\bibitem{BBM01}
Y.~Bartal, B.~Bollobas, and M.~Mendel.
\newblock A {R}amsey-type theorem for metric spaces and its applications for
  metrical task systems and related problems.
\newblock In {\em Proc. of the 42nd Ann. {IEEE} Symp. on Foundations of
  Computer Science}, pages 396--405, 2001.

\bibitem{BFR92}
Y.~Bartal, A.~Fiat, and Y.~Rabani.
\newblock Competitive algorithms for distributed data management.
\newblock In {\em Proc. of the 24th Ann. ACM Symp. on Theory of Computing},
  pages 39--50, 1992.

\bibitem{BLMN03}
Y.~Bartal, N.~Linial, M.~Mendel, and A.~Naor.
\newblock On metric {R}amsey-type phenomena.
\newblock In {\em Proc. of the 35th Ann. {ACM} Symp. on Theory of Computing},
  pages 463--472, June 2003.

\bibitem{BCL02}
W.~Bein, M.~Chrobak, and L.~L. Larmore.
\newblock The 3-server problem in the plane.
\newblock {\em Theor. Comput. Sci.}, 289(1):335--354, 2002.

\bibitem{BBCJ20}
M.~Bienkowski, J.~Byrka, C.~Coester, and \L. Je\.{z}.
\newblock Unbounded lower bound for $k$-server against weak adversaries.
\newblock In {\em Proc. of the 52nd Ann. ACM Symp. on Theory of Computing},
  pages 1165--1169, 2020.

\bibitem{BKRS92}
A.~Blum, H.~J. Karloff, Y.~Rabani, and M.~E. Saks.
\newblock A decomposition theorem and bounds for randomized server problems.
\newblock In {\em Proc. of the 33rd Ann. {IEEE} Symp. on Foundations of
  Computer Science}, pages 197--207, 1992.

\bibitem{BLS87}
A.~Borodin, N.~Linial, and M.~E. Saks.
\newblock An optimal online algorithm for metrical task systems.
\newblock In {\em Proc. of the 19th Ann. ACM Symp. on Theory of Computing},
  pages 373--382, 1987.

\bibitem{BC03}
B.~Brinkman and M.~Charikar.
\newblock On the impossibility of dimension reduction in $l_1$.
\newblock In {\em Proc. of the 44th Ann. IEEE Symp. on Foundations of Computer
  Science}, pages 514--523, 2003.

\bibitem{BCR22}
S.~Bubeck, C.~Coester, and Y.~Rabani.
\newblock Shortest paths without a map, but with an entropic regularizer.
\newblock In {\em Proc. of the 63rd Ann. IEEE Symp. on Foundations of Computer
  Science}, 2022.

\bibitem{BCLL19}
S.~Bubeck, M.~B. Cohen, J.~R. Lee, and Y.-T. Lee.
\newblock Metrical task systems on trees via mirror descent and unfair gluing.
\newblock In {\em Proc. of the 30th Ann. {ACM-SIAM} Symp. on Discrete
  Algorithms}, pages 89--97, 2019.

\bibitem{BCLLM18}
S.~Bubeck, M.~B. Cohen, Y.-T. Lee, J.~R. Lee, and A.~Madry.
\newblock $k$-server via multiscale entropic regularization.
\newblock In {\em Proc. of the 50th Ann. {ACM} Symp. on Theory of Computing},
  pages 3--16, 2018.

\bibitem{CKPV91}
M.~Chrobak, H.~J. Karloff, T.~Payne, and S.~Vishwanathan.
\newblock New results on server problems.
\newblock {\em SIAM J. Discret. Math.}, 4(2):172--181, 1991.

\bibitem{CL91}
M.~Chrobak and L.~L. Larmore.
\newblock An optimal on-line algorithm for $k$-servers on trees.
\newblock {\em SIAM J. Comput.}, 20(1):144--148, 1991.

\bibitem{C22}
C.~Coester.
\newblock Personal communication, 2022.

\bibitem{CK19}
C.~Coester and E.~Koutsoupias.
\newblock The online $k$-taxi problem.
\newblock In {\em Proc. of the 51st Ann. ACM Symp. on Theory of Computing},
  pages 1136--1147, 2019.

\bibitem{CK21}
C.~Coester and E.~Koutsoupias.
\newblock Towards the k-server conjecture: {A} unifying potential, pushing the
  frontier to the circle.
\newblock In {\em Proc. of the 48th Int'l Colloq. on Automata, Languages, and
  Programming}, pages 57:1--57:20, 2021.

\bibitem{CL19}
C.~Coester and J.~R. Lee.
\newblock Pure entropic regularization for metrical task systems.
\newblock In {\em Proc. of the 32nd Conf. on Learning Theory}, pages 835--848,
  2019.

\bibitem{FRT04}
J.~Fakcharoenphol, S.~Rao, and K.~Talwar.
\newblock A tight bound on approximating arbitrary metrics by tree metrics.
\newblock {\em J. Comput. Syst. Sci.}, 69(3):485--497, 2004.

\bibitem{FFKRRV91}
A.~{Fiat}, D.~P. {Foster}, H.~J. {Karloff}, Y.~{Rabani}, Y.~{Ravid}, and
  S.~{Vishwanathan}.
\newblock Competitive algorithms for layered graph traversal.
\newblock In {\em Proc. of the 32nd Ann. IEEE Symp. of Foundations of Computer
  Science}, pages 288--297, 1991.

\bibitem{FKLMSY91}
A.~Fiat, R.~M. Karp, M.~Luby, L.~A. McGeoch, D.~D. Sleator, and N.~E. Young.
\newblock Competitive paging algorithms.
\newblock {\em J. Alg.}, 12(4):685--699, 1991.

\bibitem{FM00}
A.~Fiat and M.~Mendel.
\newblock Better algorithms for unfair metrical task systems and applications.
\newblock In {\em Proc. of the 32nd Ann. ACM Symp. on Theory of Computing},
  pages 725--734, May 2000.

\bibitem{FRR90}
A.~Fiat, Y.~Rabani, and Y.~Ravid.
\newblock Competitive $k$-server algorithms.
\newblock In {\em Proc. of the 31st Ann. IEEE Symp. on Foundations of Computer
  Science}, pages 454--463, 1990.

\bibitem{HZ22}
Z.~Huang and H.~Zhang.
\newblock Deterministic 3-server on a circle and the limitation of canonical
  potentials, 2022.

\bibitem{ibragimov}
I.~A. Ibragimov.
\newblock A central limit theorem for a class of dependent random variables.
\newblock {\em Theory Probab. Appl.}, 8(1):83--89, 1963.
\newblock [{\it Teor. Veroyatnost. i Primenen.}, 8(1):89--94, 1963].

\bibitem{IW91}
M.~Imase and B.~M. Waxman.
\newblock Dynamic {S}teiner tree problem.
\newblock {\em SIAM J. Discret. Math.}, 4:369--384, 1991.

\bibitem{IS95}
S.~Irani and S.~S. Seiden.
\newblock Randomized algorithms for metrical task systems.
\newblock In {\em Proc. of the 4th Int'l Workshop on Algorithms and Data
  Structures}, volume 955 of {\em Lecture Notes in Computer Science}, pages
  159--170. Springer, 1995.

\bibitem{KMMO90}
A.~R. Karlin, M.~S. Manasse, L.~A. McGeoch, and S.~S. Owicki.
\newblock Competitive randomized algorithms for nonuniform problems.
\newblock {\em Algorithmica}, 11(6):542--571, 1994.

\bibitem{KRR91}
H.~J. Karloff, Y.~Rabani, and Y.~Ravid.
\newblock Lower bounds for randomized $k$-server and motion planning
  algorithms.
\newblock In {\em Proc. of the 23rd Ann. {ACM} Symp. on Theory of Computing},
  pages 278--288, 1991.

\bibitem{Kos96}
A.~P. Kosoresow.
\newblock {\em Design and Analysis of Online Algorithms for Mobile Server
  Applications}.
\newblock PhD thesis, Stanford University, 1996.

\bibitem{Kou09}
E.~Koutsoupias.
\newblock The $k$-server problem.
\newblock {\em Comput. Sci. Rev.}, 3(2):105--118, 2009.

\bibitem{KP94}
E.~Koutsoupias and C.~H. Papadimitriou.
\newblock On the $k$-server conjecture.
\newblock In {\em Proc. of the 26th Ann. ACM Symp. on Theory of Computing},
  pages 507--511, 1994.

\bibitem{KP96}
E.~Koutsoupias and C.~H. Papadimitriou.
\newblock The $2$-evader problem.
\newblock {\em Inf. Process. Lett.}, 57(5):249--252, 1996.

\bibitem{Lee18}
J.~R. Lee.
\newblock Fusible {HST}s and the randomized $k$-server conjecture.
\newblock In {\em Proc. of the 59th Ann. {IEEE} Symp. on Foundations of
  Computer Science}, pages 438--449, 2018.

\bibitem{LN04}
J.~R. Lee and A.~Naor.
\newblock Embedding the diamond graph in $l_p$ and dimension reduction in
  $l_1$.
\newblock {\em Geometric and Functional Analysis}, 14(4):745--747, 2004.

\bibitem{MMS88}
M.~S. Manasse, L.~A. McGeoch, and D.~D. Sleator.
\newblock Competitive algorithms for on-line problems.
\newblock In {\em Proc. of the 20th Ann. ACM Symp. on Theory of Computing},
  pages 322--333, 1988.

\bibitem{MN07}
M.~Mendel and A.~Naor.
\newblock Maximum gradient embeddings and monotone clustering.
\newblock In {\em Proc. of the 10th Int'l Conf. on Approximation Algorithms for
  Combinatorial Optimization Problems, and the 11th Int'l Conf. on
  Randomization and Computation}, pages 242--256, 2007.

\bibitem{NR02}
I.~Newman and Y.~Rabinovich.
\newblock A lower bound on the distortion of embedding planar metrics into
  {E}uclidean space.
\newblock In {\em Proc. of the 18th Ann. Symp. on Computational Geometry},
  pages 94--96, 2002.

\bibitem{PY89}
C.~H. Papadimitriou and M.~Yannakakis.
\newblock Shortest paths without a map.
\newblock In {\em Proc. of the 16th Int'l Colloq. on Automata, Languages and
  Programming}, pages 610--620, 1989.

\bibitem{Ram93}
H.~Ramesh.
\newblock On traversing layered graphs on-line.
\newblock In {\em Proc. of the 4th Ann. ACM-SIAM Symp. on Discrete Algorithms},
  pages 412--421, 1993.

\bibitem{Sei99}
S.~S. Seiden.
\newblock Unfair problems and randomized algorithms for metrical task systems.
\newblock {\em Inf. Comput.}, 148(2):219--240, 1999.

\bibitem{Sei01}
S.~S. Seiden.
\newblock A general decomposition theorem for the $k$-server problem.
\newblock In {\em Proc. of the 9th Ann. European Symposium on Algorithms},
  volume 2161 of {\em Lecture Notes in Computer Science}, pages 86--97.
  Springer, 2001.

\bibitem{ST84}
D.~D. Sleator and R.~E. Tarjan.
\newblock Amortized efficiency of list update and paging rules.
\newblock In {\em Proc. of the 16th Ann. ACM Symp. on Theory of Computing},
  1984.

\bibitem{Ury27}
P.~Urysohn.
\newblock Sur un espace m\'etrique universel.
\newblock {\em Bull. Sc. Math.}, 51(2):43--64, 1927.

\end{thebibliography}

\end{document}